\documentclass{article}

\usepackage[preprint]{lib/neurips_2023}

\usepackage[utf8]{inputenc} %
\usepackage[T1]{fontenc}    %
\usepackage{hyperref}       %
\usepackage{url}            %
\usepackage{booktabs}       %
\usepackage{amsfonts}       %
\usepackage{nicefrac}       %
\usepackage{microtype}      %
\usepackage{xcolor}         %
\usepackage{algorithm}
\usepackage{algorithmic}
\usepackage{lib/style}
\usepackage{lib/bold-vectors}
\usepackage{lib/bold-matrices}
\usepackage{amsmath}
\usepackage{amssymb}
\usepackage{mathtools}
\usepackage{amsthm}
\usepackage{subcaption}

\usepackage[capitalize,noabbrev]{cleveref}

\theoremstyle{plain}
\newtheorem{theorem}{Theorem}[section]

\newtheorem{lemma}[theorem]{Lemma}
\newtheorem{corollary}[theorem]{Corollary}

\theoremstyle{definition}
\newtheorem{definition}[theorem]{Definition}
\newtheorem{assumption}[theorem]{Assumption}
\theoremstyle{remark}

\author{%
  Konstantinos Ameranis \\
  Department of Computer Science\\
  University of Chicago\\
  Chicago, IL 60637 \\
  \texttt{ameranis@uchicago.edu} \\
  \AND
  Antares Chen \\
  Department of Computer Science\\
  University of Chicago\\
  Chicago, IL 60637 \\
  \texttt{achen@uchicago.edu}\\
  \And
  Adela Frances DePavia\\
  Department of Computer Science\\
  University of Chicago\\
  Chicago, IL 60637 \\
  \texttt{@uchicago.edu}
  \And
  Lorenzo Orecchia\\
  Department of Computer Science\\
  University of Chicago\\
  Chicago, IL 60637 \\
  \texttt{orecchia@uchicago.edu}
  \And
  Erasmo Tani\\
  Department of Computer Science\\
  University of Chicago\\
  Chicago, IL 60637 \\
  \texttt{etani@uchicago.edu}
}

\title{Hypergraph Diffusions and Resolvents\\ for Norm-Based Hypergraph Laplacians}

\begin{document}
\maketitle

\begin{abstract}
The development of simple and fast hypergraph spectral methods has been hindered by the lack of numerical algorithms for simulating heat diffusions and computing fundamental objects, such as Personalized PageRank vectors, over hypergraphs.
In this paper, we overcome this challenge by designing two novel algorithmic primitives. The first is a simple, easy-to-compute {\it discrete-time heat diffusion} that enjoys the same favorable properties as the discrete-time heat diffusion over graphs. This diffusion can be directly applied to speed up existing hypergraph partitioning algorithms. 

Our second contribution is the novel application of mirror descent to \textit{compute resolvents} of non-differentiable squared norms, which we believe to be of independent interest beyond hypergraph problems. Based on this new primitive, we derive the first  nearly-linear-time algorithm that simulates the discrete-time heat diffusion to approximately compute resolvents of the hypergraph Laplacian operator, which include Personalized PageRank vectors and solutions to the hypergraph analogue of Laplacian systems. Our algorithm runs in time that is linear in the size of the hypergraph and inversely proportional to the hypergraph spectral gap $\lambda_G$, matching the complexity of analogous diffusion-based algorithms for the graph version of the problem.
\end{abstract}

\section{Introduction}\label{sec:introduction}

Spectral graph methods are a fundamental tool in many machine learning tasks, including manifold learning~\cite{belkin2001laplacian, belkin2003laplacian, donoho2003hessian,li2019survey, talmon2013diffusion,yan2006graph}, clustering~\cite{kannan2004clusterings,ng2001spectral, shi2000normalized,spielman1996spectral, von2007tutorial} and network analysis~\cite{andersen2006communities,newman2004finding}. 
Within theoretical computer science, they provide very simple and fast algorithms for approximating the graph conductance $\phi_G$ of a graph $G$, a canonical graph partitioning task and a primitive subproblem in many divide-and-conquer approaches to graph algorithms~\cite{shmoys1996, vishnoi2013lx}.

Conceptually, spectral methods extract structural information about a weighted undirected graph $G=(V,E, \vw \in \R^E_{\geq 0})$ from its combinatorial Laplacian matrix $\mL$ by
studying the behavior of the heat diffusion dynamics and the related potential
$
\vx^T \mL \vx = \sum_{ij \in E} w_{ij} (\vx_i - \vx_j)^2,
$
which measures the local variance of the vector $\vx$ over $G$.
The worst-case rate of convergence of this heat diffusion to the constant vector is $(1 - \nicefrac{\lambda_G}{2})$, where $\lambda_G$ is the spectral gap of $G.$ The celebrated Cheeger's inequality~\cite{alon1984eigenvalues} connects $\lambda_G$ to the graph conductance $\phi_G$, as $\Omega(\lambda_G) \leq \phi_G \leq O(\sqrt{\lambda_G}).$ 

These elegant mathematical relations have striking algorithmic consequences in practice, as they imply that a vector $\vx$ from which the heat diffusion converges slowly can be rounded to a low conductance cut in $G$ ~\cite{andersen2006local,spielman2013local}.
More generally, the simplest spectral algorithms~\cite{vishnoi2013lx} merely simulate discrete-time heat diffusions to approximate fundamental quantities, such as the graph conductance, the effective resistance between two vertices or a personalized PageRank vector for a given seed~\cite{andersen2006local}.
The reader is referred to Appendix~\ref{sec:graph_case} for a brief review of spectral graph theory, including the relation between the heat diffusion dynamics and the natural random walk over graphs.

\paragraph{From graphs to hypergraphs}
Increasing complexity in real-world data has led to the need to model higher order relationships~\cite{ agarwalHigherOrderLearning2006, bensonHigherorderOrganizationComplex2016, catalyurekHypergraphpartitioningbasedDecompositionParallel1999,tsourakakisScalableMotifawareGraph2017}, such as co-authorship and community membership in social networks. This engendered an increased interest in hypergraph models and in transferring the powerful spectral paradigm to the hypergraph setting. 
In contrast to graphs, hypergraphs admit many choices of natural potentials and partitioning objectives \cite{veldt2022hypergraph}. In this paper, we take as a starting point the model of \citet{li2018submodular}, who provide a comprehensive framework to study hypergraph partitioning and associated potentials by considering  symmetric submodular hyperedge cut functions.

The deployment of simple and efficient hypergraph spectral algorithms  has encountered a number of challenges. In particular, the hypergraph spectral gap $\lambda_G$ is NP-hard to compute in general, so that existing methods only yield poly-logarithmic approximations based on solving semidefinite relaxations by generic convex solvers~\cite{chan2018spectral, li2018submodular, yoshida2019cheeger}. 

In contrast, polynomial-time algorithms for computing hypergraph Personalized PageRank (PPR) vectors~\cite{liu2021strongly,takai2020hypergraph}, the hypergraph generalization of graph PPR vectors~\cite{andersen2006local}, and for solving hypergraph Laplacian systems~\cite{fujiiPolynomialtimeAlgorithmsSubmodular2021} do exist. Unfortunately, all these works only provide polynomial-time or asymptotic guarantees on the convergence of their algorithms. Prior to our work, the largest obstacle in simplifying and accelerating these methods was the lack of a simple discrete-time heat-diffusion, which only exists for graphs. This led previous research endeavors to resort to either producing numerically cumbersome discretizations of the continuous-time heat diffusion~\cite{ikedaFindingCheegerCuts2019, liu2021strongly} or using generic convex solvers~\cite{takai2020hypergraph}.

\paragraph{This paper} 
We introduce a simple, efficiently computable discrete-time heat diffusion for the general class of hypergraphs described by \citet{li2018submodular} and show that it enjoys many of the favorable properties of the commonly-used graph heat diffusion.
Further, we leverage the form of our heat diffusion and novel algorithmic insights to approximately compute hypergraph PPR vectors and solve hypergraph Laplacian systems in time comparable to that required by classical algorithms for the graph analogues of these problems.
We describe our technical contributions in detail in Section~\ref{sec:contributions}, after introducing the relevant mathematical background. 

\paragraph{Other related work on hypergraph diffusions} A number of recent works pursue local hypergraph partitioning by flow-based algorithms~\cite{fountoulakisLocalHyperFlowDiffusion2021,ibrahimLocalHypergraphClustering2020,veldt2020minimizing}. Their approach requires solving local variants of the hypergraph $s$-$t$ maximum flow problem~\cite{lawlerComputingMaximalPolymatroidal1982} and is unrelated to the hypergraph spectral approach, which is based on the quadratic potentials in Equation~\ref{eq:Li-Milenkovic-potential}.
More relatedly, Macgregor and Sun~\cite{macgregorFindingBipartiteComponents2021} propose a continuous-time process to approximate a hypergraph analogue of the maximum cut problem. %
In contrast with our work, their paper does not provide an efficient way to simulate their dynamics and only yields polynomial-time algorithms.

\section{Background: spectral hypergraph theory}

For mathematical notation, definitions and background, see Section~\ref{sec:appendix-preliminaries} of the supplementary material. A weighted hypergraph $G=(V,E,\vw)$ is a collection of vertices $V$ and hyperedges $E \subseteq 2^V$ with non-negative hyperedge weights $\vw \in \R^{E}_{\geq 0}.$ 
The degree $\deg(i) \defeq \sum_{h \ni i} w_h$ of a vertex $i \in V$ is the sum of the weights of hyperedges incident to $i.$ We denote by $\mD \in \R^{V\times V}$ the diagonal matrix of degrees of $G$.

\paragraph{Cuts in graphs and hypergraphs} In graphs, the cost of a cut $(S, V\setminus S)$ for $S\subseteq V$ is the sum of the weights of the edges between $S$ and $V\setminus S$. However, for hypergraphs there is not a unique way of defining a similar partitioning objective, as there are multiple ways for a hyperedge with more than two vertices to be cut by a partition $(S,V\setminus S)$: in some applications one may want the cost of the partition to depend on how each individual hyperedge is cut. Formally, there are many ways of assigning a cut function $\delta_h$ to each hyperedge $h \in E$ to define the cost $\delta_h(h \cap S)$ incurred when the cut $(S, V \setminus S)$ crosses the hyperedge $h$.

\paragraph{Symmetric submodular cut functions} \citet{li2018submodular} consider a class of hyperedge cut functions called \emph{symmetric submodular hyperedge cut functions} $\delta_h:2^h \to \R$. These are functions which satisfy three key properties. Firstly, each $\delta_h$ is \emph{submodular}. Second, $\delta_h$ is symmetric, meaning that for any $A \subseteq h$, $\delta_h(A) = \delta_h(h \setminus A)$. Finally $\delta_h(h) = \delta_h(\emptyset) = 0$. This model encompasses many cut functions that are prominent in applications, including cardinality-based cut functions~\cite{veldt2020minimizing} and the mutual information cut function~\cite{gretton2003kernel,xiong2005optimizing}.

\paragraph{(Hyper)graph potentials}
In graphs, we generalize the cost of a cut $\delta_{i,j} = \mathbbm{1}\{ij\text{ is cut}\}$ to a function on $\R^V$ as $\overline{\delta}_{ij} = |\vx_i - \vx_j|$. The \textit{global graph potential} is defined to be the weighted sum of the squares of these edge costs, and is a core quantity studied in spectral graph theory:
\[
    U(\vx) \defeq \sum_{e\in E} w_e |\vx_i - \vx_j|^2 = \vx^T \mL \vx.
\]
Note in fact that when $\vx = \vone_S$ is the zero-one indicator of a set $S$, then $U(\vx)$ exactly equals the cost of the cut $(S, V\setminus S)$.

On hypergraphs, for any submodular hyperedge cut function $\delta_h$ its \textit{Lov\'asz extension} $\overline{\delta}_h:\R^n\rightarrow \R$ allows one to move from penalizing cuts to penalizing vectors $\vx\in \R^V$. In particular, for symmetric submodular hyperedge cut functions, \citet{li2018submodular} study the corresponding hypergraph potential:
\begin{equation}\label{eq:Li-Milenkovic-potential}
    \cU_{\textrm{submod}}(\vx) = \frac{1}{2}\sum_{h\in E} w_h\overline{\delta}_h(\vx)^2.
\end{equation}
Potentials of this form enjoy several convenient properties, like being amenable to efficient optimization algorithms~\cite{baiAlgorithmsOptimizingRatio2016,eneRandomCoordinateDescent2015}, and satisfying a general hypergraph analogue of Cheeger inequality for graphs~\cite{yoshida2019cheeger}.

\paragraph{Example: the standard hyperedge cut function}A canonical example of a cut function in this class is the function given by:
\[
    \delta^{\infty}_h(S \cap h) \defeq \min\{1, |S \cap h|, |(V\setminus S) \cap h|\}
\]
which gives a cost of one unit to every hyperedge which is cut, regardless of which vertices fall on which side of the cut. \citet{louis2015hypergraph} and \citet{chan2018spectral} show that the corresponding standard potential is
\begin{equation}\label{eq:infinity_energy_functional}
    \cU_{\infty}(\vx) \defeq \frac{1}{2} \sum_{h\in E} w_h \min_{u\in \R}\norm{\vxh - u \1_h}_\infty^2 = \frac{1}{4}\sum_{h\in E} w_h \max_{i,j\in h}(\vx_i-\vx_j)^2,
\end{equation}
where $\vx_h$ is the restriction of $\vx \in \R^V$ to $\R^h.$ They also establish a Cheeger's inequality and show that this choice of hyperedge cut function captures vertex-based graph partitioning problems~\cite{louis2016approximation}. 

\section{Our contributions}\label{sec:contributions}

In this paper, we directly tackle the computational and technical challenges in developing hypergraph spectral methods highlighted in the introduction. We establish connections between key hypergraph primitives--such as heat diffusions, personalized PageRank vectors, and hypergraph Laplacian systems--and hypergraph potentials. We then develop novel fast algorithms for a large class of hypergraph optimization problems, involving potentials of the form:
\begin{equation}\label{eq:our-setting}
    \cU(\vx) \defeq \frac{1}{2}\sum_{h\in E}w_h \min_{u\in \R}\norm{\vxh - u \1_h}_h^2,
\end{equation}
where each $\|\cdot\|_h$ is a norm over $\R^h.$ In particular, this class contains all quadratic potentials associated with the symmetric submodular cut functions of the form in Equation~\eqref{eq:Li-Milenkovic-potential} (see Lemma~\ref{lem:lovasz-is-semi-norm}).
Throughout the paper, we let $\lambda_G$ be the Poincaré constant of $G$:
\begin{equation}\label{eq:definition-lambda-G}
    \lambda_G \defeq \min_{\vx \perp \mD \1}  {\cU(\vx)\over\;\; {1\over 2}\norm{\vx}_{\mD}^2},
\end{equation}
which is a direct generalization (up to constants) of the notion of spectral gap employed by previous works.
Within this setup, we make the following main technical contributions:
\begin{itemize}[leftmargin=*]
\item  in Section~\ref{sec:heat_diffusion}, we introduce a simple, easy-to-compute {\it heat diffusion over hypergraphs}. We establish fast convergence of this heat diffusion in both continuous- and discrete-time. We show that the discrete-time  diffusion has the same favorable properties of the discrete-time heat diffusion for graphs, which immediately allows its application in the context of local hypergraph partitioning. %
In Section~\ref{sec:experiments}, we provide experiments showing that our hypergraph heat diffusion outperforms graph-based heat diffusions in semi-supervised manifold learning tasks.
\item in Section~\ref{sec:resolvent}, we introduce the {\it hypergraph resolvent problem}, which subsumes both the computation of hypergraph PPR vectors~\cite{takai2020hypergraph} and the solution of hypergraph Laplacian systems~\cite{fujiiPolynomialtimeAlgorithmsSubmodular2021}. 
We present the first nearly-linear-time algorithm to approximately solve this task by specializing a novel optimization primitive, which is our next contribution in Section~\ref{sec:optimization}.
Our algorithm iteratively simulates the discrete-time heat diffusion over the graph. Analogously to its graph counterparts,  its runs in linear time in the size of the hypergraph, with a $\lambda_G^{-1}$ dependence on the spectral gap, which is necessary for the heat diffusion to mix. In Section~\ref{sec:empirical-comparison} in the Supplementary Material, we present an empirical comparison between our algorithm and previous fast heuristics, showing the superiority of our method.
\item in Section~\ref{sec:optimization}, we develop a {\it novel first-order optimization algorithm} for problems involving the minimization of sums of non-differentiable squared norms, such as the general dissipative potentials of Equation~\ref{eq:our-setting}. We believe this result to be of independent interest beyond the study of hypergraph algorithms.

\end{itemize}

\paragraph{Paper organization} All references to sections with alphabetic indexes refer to the Supplementary Material. All proofs of claimed statements are included in the Supplementary Material. A short guide to our notation is included in Section~\ref{sec:appendix-preliminaries}.

\section{Hypergraph potential and hypergraph Laplacian}\label{sec:preliminaries}
In this section, we establish key properties of the hypergraph potentials in Equation~\ref{eq:our-setting}. For the rest of this section we assume familiarity with properties of submodular and convex functions: for further background on these topics, see Section~\ref{sec:appendix-preliminaries} of the supplementary material.

We consider hypergraph potentials $\cU(\vx)$ of the form of Equation~\ref{eq:our-setting}, which are necessarily convex, as they are the sum of squared semi-norms. The following lemma, which is proved in Section~\ref{sec:preliminaries-proof} of the supplementary material, implies that the class described by Equation~\ref{eq:our-setting} also includes all potentials arising from symmetric submodular cut functions in the sense of \citet{li2018submodular}:
\begin{lemma}\label{lem:lovasz-is-semi-norm}
For any symmetric submodular cut function $\delta_h: 2^h \to \R$, the Lov\'asz extension $\bar{\delta_h}: \R^h \to \R$ takes the form
$$
\bar{\delta}_h(\vx) = \min_u \|\vx_h - u \vone_h\|_h,
$$
for some norm $\|\cdot \|_h.$ Hence, the potential $\frac{1}{2}\sum_{h \in E} w_h \bar{\delta}^2_h(\vx)$ takes the form of Equation~\ref{eq:our-setting}.
\end{lemma}

A key set of interest used to define diffusions over hypergraphs is the \textit{subdifferential} of the energy functional $\cU(\vx)$, which we denote $\cL(\vx) \defeq \partial \cU(\vx)$. 
By standard properties of subgradients, we can characterize $\cL(\vx)$ explicitly as %
\begin{equation}\label{eq:laplacian}
    \cL(\vx) = \left\{\sum_{h\in E} w_h \left(\min_{u\in \R}\norm{\vxh - u\1_h}_h\right) \vy_h) \right\},
\textrm{ where } 
\vy_h \in \argmax_{\substack{\vy \perp \vone\\ \norm{\vy}_{h,*} \leq 1}} \ip{\vy, \vx_h}.
\end{equation}
For the standard potential $\cU_\infty$, we give an example of the computation of $\cL(\vx)$ in Figure~\ref{fig:single-edge-flow}.

The set-valued operator $\cL(\cdot)$ is the nonlinear analogue
of the graph Laplacian operator $\mL$. Indeed, when a hyperedge $h=\{i,j\}$ has cardinality $2$ and its associated norm is $\ell_2$ or $\ell_\infty$, its contribution to $\cL(\vx)$ is $\frac{w_h}{2}(\vx_i - \vx_j)(\ve_i - \ve_j)$. In particular, 
in the case of a two-uniform hypergraph with $\norm{\cdot}_h = \norm{.}_2$ for all $h$, we obtain
$\cL(\vx) = \frac{1}{2}\mL\vx$ for the combinatorial graph Laplacian $\mL$.

\paragraph{Norms bounds}
We make the following normalization assumption. Notice that this causes no loss in generality as any additional scaling can be incorporated into the hypergraph weight $w_{h}.$ 
\begin{assumption}\label{assumption}
     $\norm{\vx}_h \leq \norm{\vx}_2$ for every $h \in E, \vx \in \R^V.$
\end{assumption}

With this assumption, we can show the following bounds relating the squared semi-norm $\cU(\vx)$ with the squared degree norm. The proof appears in Section~\ref{sec:preliminaries-proof} of the Supplementary Material.
\begin{lemma}\label{lemma:eigenvalue-type-bounds}
For every connected hypergraph $G$, $\lambda_G > 0$. For any hypergraph and any choice of $\{\norm{\cdot}_h\}$ norms, for all $\vx \perp_{\mD} \1$ the potential energy functional $\cU$  satisfies:
\begin{equation}\label{eq:eigenvalue-type-bounds}
    \lambda_G \leq \frac{\cU(\vx)}{{1\over 2}\norm{\vx}^2_{\mD}} \leq 1.
\end{equation}
\end{lemma}
In graphs, the spectral gap of the Laplacian matrix controls the behavior of heat diffusions. In the next section, we will use Lemma~\ref{lemma:eigenvalue-type-bounds} to show that $\lambda_G$ controls nonlinear processes on hypergraphs in an analogous way.

\section{Heat diffusion on hypergraphs}\label{sec:heat_diffusion}

The \textit{heat equation} is a fundamental partial differential equation:
\[
    \dot{u} = \Delta u  \quad u:\R^n \times [0,\infty)\rightarrow \R,
\]
where $\dot{u}$ denotes the partial-derivative of $u$ with respect to time and $\Delta$ denotes the Laplace operator.
The heat diffusion on a graph is an analogous ordinary differential equation, where the Laplace operator is replaced by a suitable normalization of the combinatorial Laplacian $\mL \in \R^{V\times V}$:
\begin{equation}\label{eq:graph-heat-diffusion}
\dot{\vx}(t) = - \mD^{-1}\mL \vx(t)  \quad \vx:[0,\infty)\rightarrow \R^V.
\end{equation}
Equivalently, this diffusion performs an  infinitesimal averaging of the value of $\vx_i$ at each vertex $i$ with the values of neighbors of $i$. For the graph case, a common alternative definition is given by considering the change of basis $\vp(t) = \mD \vx(t).$ Then, $\vp(t)$ can be interpreted as the marginals of the continuous-time random walk over the graph. %
Crucially, the heat diffusion in Equation~\ref{eq:graph-heat-diffusion} arises as the gradient flow of the graph potential $\cU(\vx) = \nicefrac{1}{2} \cdot \vx^T \mL \vx$ with respect to the norm $\mD.$

The difficulty in considering the corresponding heat diffusion on hypergraphs lies in the fact that the  potential $\cU(\vx)$ is now no longer differentiable, but rather only \textit{sub}differentiable. Heat diffusion across hypergraphs is thus described by the following subgradient flow:
\[
    \dot{\vx} \in -\mD^{-1}\cL(\vx).
\]
Moreover, because of the non-linearity of $\cL$, this evolution cannot be related to the marginals of a Markov chain over the hypergraph.
While \citet{ikedaFindingCheegerCuts2019} have shown the existence and uniqueness of the solution to this subgradient inclusion for general convex potentials, we provide a slightly stronger result by fully characterizing such unique solution as an ordinary differential equation.

To do so, we generalize the gradient-flow characterization of the graph heat diffusion to hypergraphs and cast the heat diffusion as the gradient flow of the potential $\cU(\vx)$ under the degree norm $\| \cdot \|_{\mD}$. We can now rely on well-established results on the existence and uniqueness of gradient flows.

\begin{theorem}[Unique solution to the subgradient flow]\label{thm:characterizing_subgradient_flow}
    For every $\vx_0\in\R^n$, there exists a unique solution $\vx(t) \in C([0,\infty);\R^n)$ with $\dot{\vx}(t)\in L^{\infty}(0,\infty;\R^n)$ such that
    \[
        \begin{cases}
            \vx(0)=  \vx_0,\\
            \dot{\vx}(t) \in -\mD^{-1}\cL(\vx(t)) \text{ almost everywhere for } t\geq 0.
        \end{cases}
    \]
    Moreover, at \textit{all} times $t > 0$ the right-derivative exists and can be exactly characterized:
    \[
        \lim_{h\rightarrow 0}\frac{1}{h}(\vx(t+h)-\vx(t)) = -\cLD(\vx(t)),
    \]
    where we define the operator
    \[
        \cLD(\vx) \defeq \argmin_{y \in \cL(\vx)}\norm{\vy}_{\mD^{-1}}.
    \]
\end{theorem}

In words, the unique solution at $\vx$ tends to follow the subgradient element $\cLD(\vx)$ of minimum $\norm{\cdot}_{\mD^{-1}}$-norm. We note that the $\argmin$ is unique, since $\norm{\cdot}_{\mD^{-1}}$ is a strongly-convex function. One should think of $\cLD(\vx)$ as the direction of steepest descent for the function $\cU$ at $\vx$ with respect to the $\norm{\cdot}_{\mD}$-norm \cite{brezis1972operateurs, evans2010partial}. This fact is computationally significant as it will yield a principled way to select elements of the subdifferential in the discrete-time diffusion below.

\paragraph{Properties of hypergraph heat diffusion}

Like the graph heat diffusion, 
the hypergraph heat diffusion preserves the mean of $\vx(t)$ with respect to the degree measure and has the constant vector as a fixed point. More strongly, for a connected hypergraph, the heat diffusion started at $\vx$  converges to $\pi(\vx) \defeq \nicefrac{\ip{\vx,\vone}_{\mD}}{\|\vone\|^2_{\mD}} \cdot \vone,$ the $\mD$-projection of $\vx$ onto $\vone.$  
Just as in the graph case, the worst-case convergence is controlled by the Poincar\'e constant $\lambda_G$. The following theorem is proved in Section~\ref{appendix:missing-proofs}.

\begin{theorem}\label{thm:cts-time-diffusion-convergence}
For any $T \in [0,\infty)$ and any initial point $\vx_0\in \R^n$, suppose $\vx(t)$ is continuous on $[0,\infty)$ and satisfies $\dot{\vx}(t) \in -\mD^{-1}\cL(\vx(t))$  a.e. on $[0,T]$. Then, the following convergence properties hold: 
\begin{align*}
\mathrm{instantaneous}&:  \quad \quad
 \frac{d}{dt}\left(\frac{1}{2}\norm{\vx(t) -\pi(\vx_0)}\right) = \; -2\cU(\vx(t)) \quad \text{ a.e. on } [0,T],\\
 \mathrm{aggregate}&: \qquad \qquad 
  \frac{1}{2}\norm{\vx(t) -\pi(\vx_0)}^2_{\mD} \; \;\, \leq\; \frac{1}{2}\norm{\vx_0 -\pi(\vx_0)}^2_{\mD} \cdot \exp\left(-2\lambda_G \cdot t\right) \  \forall t \in [0,T].
\end{align*}
\end{theorem}
For connected hypergraphs, the previous theorem implies that the heat diffusion converges to the degree-weighted average for each vertex, with variance that decreases geometrically with rate $\lambda_G.$ %
We establish more properties of the continuous-time heat diffusion, including a maximum principle that is characteristic of diffusions, in Section~\ref{sec:properties} in the Supplementary Material.

\subsection{Discrete-time evolution}

\begin{figure*}[t]
\begin{center}
\tikzset{every picture/.style={line width=0.75pt}} %

\begin{tikzpicture}[x=0.75pt,y=0.75pt,yscale=-1,xscale=1]
\draw  [color={rgb, 255:red, 0; green, 98; blue, 213 }  ,draw opacity=1 ][fill={rgb, 255:red, 242; green, 247; blue, 255 }  ,fill opacity=1 ] (395.23,94.7) .. controls (412.23,77.7) and (492.23,76.7) .. (505.23,95.7) .. controls (518.23,114.7) and (517.23,172.7) .. (505.23,191.7) .. controls (493.23,210.7) and (407.23,207.7) .. (394.23,194.7) .. controls (381.23,181.7) and (378.23,111.7) .. (395.23,94.7) -- cycle ;
\draw  [color={rgb, 255:red, 0; green, 98; blue, 213 }  ,draw opacity=1 ][fill={rgb, 255:red, 242; green, 247; blue, 255 }  ,fill opacity=1 ] (230.23,95.7) .. controls (247.23,78.7) and (327.23,77.7) .. (340.23,96.7) .. controls (353.23,115.7) and (352.23,173.7) .. (340.23,192.7) .. controls (328.23,211.7) and (242.23,208.7) .. (229.23,195.7) .. controls (216.23,182.7) and (213.23,112.7) .. (230.23,95.7) -- cycle ;
\draw  [color={rgb, 255:red, 0; green, 98; blue, 213 }  ,draw opacity=1 ][fill={rgb, 255:red, 242; green, 247; blue, 255 }  ,fill opacity=1 ] (56.23,95.7) .. controls (73.23,78.7) and (153.23,77.7) .. (166.23,96.7) .. controls (179.23,115.7) and (178.23,173.7) .. (166.23,192.7) .. controls (154.23,211.7) and (68.23,208.7) .. (55.23,195.7) .. controls (42.23,182.7) and (39.23,112.7) .. (56.23,95.7) -- cycle ;
\draw  [color={rgb, 255:red, 0; green, 0; blue, 0 }  ,draw opacity=1 ][fill={rgb, 255:red, 229; green, 240; blue, 255 }  ,fill opacity=1 ][line width=0.75]  (60.69,105.81) .. controls (60.69,102.89) and (63.05,100.53) .. (65.97,100.53) .. controls (68.89,100.53) and (71.25,102.89) .. (71.25,105.81) .. controls (71.25,108.73) and (68.89,111.09) .. (65.97,111.09) .. controls (63.05,111.09) and (60.69,108.73) .. (60.69,105.81) -- cycle ;
\draw  [color={rgb, 255:red, 0; green, 0; blue, 0 }  ,draw opacity=1 ][fill={rgb, 255:red, 229; green, 240; blue, 255 }  ,fill opacity=1 ][line width=0.75]  (148.81,106.36) .. controls (148.81,103.44) and (151.18,101.08) .. (154.1,101.08) .. controls (157.01,101.08) and (159.38,103.44) .. (159.38,106.36) .. controls (159.38,109.28) and (157.01,111.64) .. (154.1,111.64) .. controls (151.18,111.64) and (148.81,109.28) .. (148.81,106.36) -- cycle ;
\draw  [color={rgb, 255:red, 0; green, 0; blue, 0 }  ,draw opacity=1 ][fill={rgb, 255:red, 0; green, 112; blue, 245 }  ,fill opacity=1 ][line width=0.75]  (60.78,183.96) .. controls (60.78,181.04) and (63.15,178.67) .. (66.06,178.67) .. controls (68.98,178.67) and (71.35,181.04) .. (71.35,183.96) .. controls (71.35,186.87) and (68.98,189.24) .. (66.06,189.24) .. controls (63.15,189.24) and (60.78,186.87) .. (60.78,183.96) -- cycle ;
\draw  [color={rgb, 255:red, 0; green, 0; blue, 0 }  ,draw opacity=1 ][fill={rgb, 255:red, 229; green, 240; blue, 255 }  ,fill opacity=1 ][line width=0.75]  (234.69,104.81) .. controls (234.69,101.89) and (237.05,99.53) .. (239.97,99.53) .. controls (242.89,99.53) and (245.25,101.89) .. (245.25,104.81) .. controls (245.25,107.73) and (242.89,110.09) .. (239.97,110.09) .. controls (237.05,110.09) and (234.69,107.73) .. (234.69,104.81) -- cycle ;
\draw  [color={rgb, 255:red, 0; green, 0; blue, 0 }  ,draw opacity=1 ][fill={rgb, 255:red, 229; green, 240; blue, 255 }  ,fill opacity=1 ][line width=0.75]  (322.81,105.36) .. controls (322.81,102.44) and (325.18,100.08) .. (328.1,100.08) .. controls (331.01,100.08) and (333.38,102.44) .. (333.38,105.36) .. controls (333.38,108.28) and (331.01,110.64) .. (328.1,110.64) .. controls (325.18,110.64) and (322.81,108.28) .. (322.81,105.36) -- cycle ;
\draw  [color={rgb, 255:red, 0; green, 0; blue, 0 }  ,draw opacity=1 ][fill={rgb, 255:red, 0; green, 52; blue, 114 }  ,fill opacity=1 ][line width=0.75]  (323.45,182.88) .. controls (323.45,179.96) and (325.82,177.6) .. (328.74,177.6) .. controls (331.65,177.6) and (334.02,179.96) .. (334.02,182.88) .. controls (334.02,185.8) and (331.65,188.16) .. (328.74,188.16) .. controls (325.82,188.16) and (323.45,185.8) .. (323.45,182.88) -- cycle ;
\draw  [color={rgb, 255:red, 0; green, 0; blue, 0 }  ,draw opacity=1 ][fill={rgb, 255:red, 0; green, 112; blue, 245 }  ,fill opacity=1 ][line width=0.75]  (234.78,181.96) .. controls (234.78,179.04) and (237.15,176.67) .. (240.06,176.67) .. controls (242.98,176.67) and (245.35,179.04) .. (245.35,181.96) .. controls (245.35,184.87) and (242.98,187.24) .. (240.06,187.24) .. controls (237.15,187.24) and (234.78,184.87) .. (234.78,181.96) -- cycle ;
\draw  [color={rgb, 255:red, 0; green, 0; blue, 0 }  ,draw opacity=1 ][fill={rgb, 255:red, 103; green, 175; blue, 255 }  ,fill opacity=1 ][line width=0.75]  (399.69,103.81) .. controls (399.69,100.89) and (402.05,98.53) .. (404.97,98.53) .. controls (407.89,98.53) and (410.25,100.89) .. (410.25,103.81) .. controls (410.25,106.73) and (407.89,109.09) .. (404.97,109.09) .. controls (402.05,109.09) and (399.69,106.73) .. (399.69,103.81) -- cycle ;
\draw  [color={rgb, 255:red, 0; green, 0; blue, 0 }  ,draw opacity=1 ][fill={rgb, 255:red, 103; green, 175; blue, 255 }  ,fill opacity=1 ][line width=0.75]  (487.81,104.36) .. controls (487.81,101.44) and (490.18,99.08) .. (493.1,99.08) .. controls (496.01,99.08) and (498.38,101.44) .. (498.38,104.36) .. controls (498.38,107.28) and (496.01,109.64) .. (493.1,109.64) .. controls (490.18,109.64) and (487.81,107.28) .. (487.81,104.36) -- cycle ;
\draw  [color={rgb, 255:red, 0; green, 0; blue, 0 }  ,draw opacity=1 ][fill={rgb, 255:red, 0; green, 98; blue, 221 }  ,fill opacity=1 ][line width=0.75]  (488.45,181.88) .. controls (488.45,178.96) and (490.82,176.6) .. (493.74,176.6) .. controls (496.65,176.6) and (499.02,178.96) .. (499.02,181.88) .. controls (499.02,184.8) and (496.65,187.16) .. (493.74,187.16) .. controls (490.82,187.16) and (488.45,184.8) .. (488.45,181.88) -- cycle ;
\draw  [color={rgb, 255:red, 0; green, 0; blue, 0 }  ,draw opacity=1 ][fill={rgb, 255:red, 0; green, 112; blue, 245 }  ,fill opacity=1 ][line width=0.75]  (399.78,181.96) .. controls (399.78,179.04) and (402.15,176.67) .. (405.06,176.67) .. controls (407.98,176.67) and (410.35,179.04) .. (410.35,181.96) .. controls (410.35,184.87) and (407.98,187.24) .. (405.06,187.24) .. controls (402.15,187.24) and (399.78,184.87) .. (399.78,181.96) -- cycle ;
\draw  [dash pattern={on 4.5pt off 4.5pt}]  (246.98,111.31) -- (328.74,182.88) ;
\draw [shift={(245.47,109.99)}, rotate = 41.2] [color={rgb, 255:red, 0; green, 0; blue, 0 }  ][line width=0.75]    (10.93,-3.29) .. controls (6.95,-1.4) and (3.31,-0.3) .. (0,0) .. controls (3.31,0.3) and (6.95,1.4) .. (10.93,3.29)   ;
\draw  [dash pattern={on 4.5pt off 4.5pt}]  (328.11,112.64) -- (328.74,182.88) ;
\draw [shift={(328.1,110.64)}, rotate = 89.49] [color={rgb, 255:red, 0; green, 0; blue, 0 }  ][line width=0.75]    (10.93,-3.29) .. controls (6.95,-1.4) and (3.31,-0.3) .. (0,0) .. controls (3.31,0.3) and (6.95,1.4) .. (10.93,3.29)   ;
\draw  [color={rgb, 255:red, 0; green, 0; blue, 0 }  ,draw opacity=1 ][fill={rgb, 255:red, 0; green, 52; blue, 114 }  ,fill opacity=1 ][line width=0.75]  (149.45,183.88) .. controls (149.45,180.96) and (151.82,178.6) .. (154.74,178.6) .. controls (157.65,178.6) and (160.02,180.96) .. (160.02,183.88) .. controls (160.02,186.8) and (157.65,189.16) .. (154.74,189.16) .. controls (151.82,189.16) and (149.45,186.8) .. (149.45,183.88) -- cycle ;

\draw (99,61.4) node [anchor=north west][inner sep=0.75pt]    {$w_{h}$};
\draw (247.97,95) node [anchor=north west][inner sep=0.75pt]  [font=\small]  {$-\frac{3w_{h}}{8}$};
\draw (289,95) node [anchor=north west][inner sep=0.75pt]  [font=\small]  {$-\frac{3w_{h}}{8}$};
\draw (67.97,110) node [anchor=north west][inner sep=0.75pt]    {$-1$};
\draw (125.97,110) node [anchor=north west][inner sep=0.75pt]    {$-1$};
\draw (75,170) node [anchor=north west][inner sep=0.75pt]    {$1$};
\draw (135,170) node [anchor=north west][inner sep=0.75pt]    {$2$};
\draw (14,70.4) node [anchor=north west][inner sep=0.75pt]    {$( a)$};
\draw (192,70.4) node [anchor=north west][inner sep=0.75pt]    {$( b)$};
\draw (415,100) node [anchor=north west][inner sep=0.75pt]  [font=\small]  {$-\frac{5}{8}$};
\draw (361,70.4) node [anchor=north west][inner sep=0.75pt]    {$( c)$};
\draw (289,175) node [anchor=north west][inner sep=0.75pt]  [font=\small]  {$+\frac{3w_{h}}{4}$};
\draw (249.97,175) node [anchor=north west][inner sep=0.75pt]    {$0$};
\draw (472,168) node [anchor=north west][inner sep=0.75pt]  [font=\small]  {$\frac{5}{4}$};
\draw (235,190) node [anchor=north west][inner sep=0.75pt]  [font=\tiny] [align=left] {(No change)};
\draw (465,100) node [anchor=north west][inner sep=0.75pt]  [font=\small]  {$-\frac{5}{8}$};
\draw (412.97,170) node [anchor=north west][inner sep=0.75pt]    {$1$};

\end{tikzpicture}
\vskip -0.2in
    \caption{
    $(a)$ A hypergraph consisting of a single hyperedge $h$, with weight $w_h$, and heat distribution $\vx_t = (-1,-1,1,2)$. 
    $(b)$ In Equation~\ref{eq:laplacian}, we have $\norm{\vxh - u\1_h}_h= 3/2$ and\;$\vy_h = (-1/2+\alpha, -\alpha, 0 , 1/2)$ for $\alpha \in (0,1/2)$, giving $\cL(\vx_t) =\{w_h(-3/4+3\alpha/2, -3\alpha/2, 0 , 3/4)\mid \alpha\in(0,1/2)\}$.
    The vector $ \cLD(\vx_t) = (-3w_h/8,-3w_h/8,0,+3w_h/4)$ is the minimum $\norm{\cdot}_{\mD^{-1}}$-norm vector in this set. This can be interpreted as a flow on $h$, where the value of $\cLD(\vx_t)(i)$ is the amount of flow flowing into vertex $i$ through the hyperedge. $(c)$ The discrete-time heat diffusion then computes $ \vx_{t+1} = \vx_t - {1\over w_h} \cL^{\mD}(\vx_t)$ giving $\vx_{t+1} = (-5/8,-5/8,1,5/4)$.}
    \label{fig:single-edge-flow}
\end{center}
\vspace{-0.5cm}
\end{figure*}
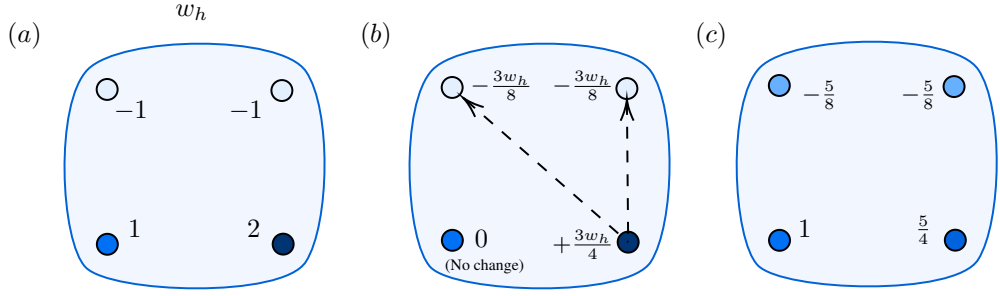

In practice, continuous-time heat diffusions are computationally cumbersome, as witnessed by the ongoing difficulty~\cite{moler1978nineteen, moler2003nineteen} in approximating the action of the matrix exponential. However, this is not an issue for graphs, as the continuous-time heat diffusion can be replaced for all algorithmic purposes with the discrete-time heat diffusion 
\begin{equation} \label{eq:step-length-graph}
\vx_{t+1} = \vx_{t} - {1\over 2} \mD^{-1} \mL\vx_{t} = \left(\vI -{1\over 2} \mD^{-1} \mL\right) \vx_t
\end{equation}
While the latter does not strictly approximate the continuous-time process $\vx(t)$, it shares with it the properties of Theorem~\ref{thm:cts-time-diffusion-convergence}, which are crucial in applications. %

For the main contribution of this section, we show that the forward-Euler discretization of $\vx(t)$, with the {\it same constant step length}\footnote{Step lengths in Equations~\ref{eq:step-length-graph} and ~\ref{eq:0-s-discrete-update-rule} differ because $\cL = \frac{1}{2} \mL$ for graphs.} as in the graph case, satisfies the two properties above, overcoming the challenge of defining a discrete-time evolution for non-differentiable hypergraph potentials, which was previously raised in \cite{chan2018spectral}.
We now describe our discrete-time dynamics and prove the required properties. Fixing a starting point $\vx_0$, we compute a sequence $\{\vx_t\}_{t\in \N}$ by repeatedly applying:
\begin{equation}\label{eq:0-s-discrete-update-rule}
    \vx_{t+1} = \vx_t - 
    \mD^{-1}\cLD(\vx_t).
\end{equation}
 We show that the convergence behavior of this discrete-time process matches that of the continuous-time guarantee in Theorem~\ref{thm:cts-time-diffusion-convergence}. Proofs are in Section~\ref{appendix:missing-proofs} of the Supplementary Material: 

\begin{theorem}\label{thm:discrete-time-diffusion}
Consider a sequence described by Equation~\eqref{eq:0-s-discrete-update-rule}, initialized at some $\vx_0$, then $\forall t\in\mathbb{N}_{\geq 0}$, $\forall t\in\mathbb{N}_{\geq 0}$. Then, the following convergence properties hold: 
\begin{align*}
\mathrm{``instantaneous"}&:  \quad \quad
 \norm{\vx_{t+1} - \pi(\vx_0)}^2_{\mD} \leq \norm{\vx_t-\pi(\vx_0)}^2_{\mD} - 2\cU(\vx_t)\\
 \mathrm{aggregate}&: \qquad  
  \norm{\vx_t - \pi(\vx_0)}_{\mD}^2 \leq (1-\lambda_G)^t \norm{\vx_0 - \pi(\vx_0)}^2_{\mD}.
\end{align*}
\end{theorem}
In Section~\ref{sec:local} in the Supplementary Material, we show how this theorem can be directly applied to improve the result of \citet{ikedaFindingCheegerCuts2019} on local hypergraph partitioning. %

\section{Resolvents on hypergraphs}\label{sec:resolvent}
Non-linear resolvent operators are a fundamental object of study in dynamical systems and semigroup theory~\cite{bauschkeConvexAnalysisMonotone2011, brezis1972operateurs, evans2010partial}.
For a parameter $\lambda \geq 0,$ the resolvent operator $\cR_{\lambda}(\vs)$ of the hypergraph  potential $\cU$ maps a seed vector $\vs \in \R^V$ to the set of minimizers of the following optimization problem:
\begin{equation}\label{eq:resolvent}
\min
_{\vx \in \R^V} 
\cU(\vx) + \frac{\lambda}{2} \|\vx\|^2_{\mD} - \ip{\vs,\vx}.
\end{equation}
Letting $\cL$ be the subgradient of $\cU$, the optimality conditions of~\eqref{eq:resolvent} yield the following characterization:
\[
    \vx \in \cR_{\lambda}(\vs) \iff \vs - \lambda \cdot \mD \vx \in \cL(\vx).
\]

For the case $\lambda = 0$ and $\vs \perp \vone$,  this is the hypergraph analogue of the graph Laplacian system $\mL \vx = \vs$, whose solutions yield the voltages of electrical flows over graphs~\cite{vishnoi2013lx}.
The hypergraph version was previously studied by \citet{fujiiPolynomialtimeAlgorithmsSubmodular2021}, who gave polynomial-time algorithms for approximating the solution $\vx.$
When $\lambda > 0,$ the resolvent $\cR_{\lambda}(\vs)$ contains a unique element, as the objective in Problem 8 becomes strongly convex. The following lemma shows that this element equals, up to scaling, the hypergraph PPR vector studied by \citet{takai2020hypergraph} and \citet{liu2021strongly}  in the context of local hypergraph partitioning. 
\begin{lemma}\label{lem:hypergraph-pagerank-fixed-point} For teleportation parameter $\alpha \in (0,1)$ and seed vector $\vs \in \R^V$,
define the hypergraph PPR vector $\vp_{\alpha}(\vs)$ as the unique solution to the equation 
$
\vp_{\alpha}(\vs) + \frac{1-\alpha}{2\alpha} \cL(\mD^{-1} \vp_{\alpha}(\vs)) \ni \vs.
$
Then, $\mD^{-1} \vp_{\alpha}(\vs)$ is the unique element in $ \cR_{\nicefrac{2\alpha}{1-\alpha}}((\nicefrac{2\alpha}{1-\alpha})\vs).
$
\end{lemma}

In the graph case, the resolvent $\cR_{\lambda}(\vs)$ is given by $(\lambda \mD + \mL)^{-1} \vs$. The following geometric-series expansion shows that this resolvent can be written as a conic combination of  discrete-time heat diffusion started at $\mD^{-1} \vs:$
\begin{equation}\label{eq:expansion}
(\lambda \mD + \mL)^{-1} \vs 
= \frac{1}{\lambda+2} \cdot  \sum_{k=0}^\infty \left(\frac{2}{\lambda+2}\right)^k \cdot \left(\mI - \frac{\mD^{-1} \mL}{2}\right)^k \mD^{-1} \vs
\end{equation}

Gradient descent methods for approximating the graph resolvent can be shown to be equivalent to computing a truncation of this series, where $O(\nicefrac{1}{\lambda_G} \cdot \log(\nicefrac{1}{\epsilon}))$ terms suffice to obtain a multiplicative $\epsilon$-approximation to the optimum of Problem~\ref{eq:resolvent}~\cite{vishnoi2013lx}.
Unfortunately, such series expansion crucially depends on the linearity of  the graph Laplacian $\mL$ and breaks down for non-linear operators.
At the same time, standard applications of first-order methods~\cite{nemirovskij1983problem, nesterov} fail to give algorithms for approximately solving Problem~\ref{eq:resolvent}, because the non-differentiable quadratic objective is both non-smooth and non-Lipschitz.
In Theorem~\ref{thm:optimization} i the next section, we analyze Algorithm~\ref{alg:opt}, a specialized first-order methods for this challenge. By setting $\mR = \mD$ as a choice of prox-generating function in Algorithm 1~\ref{alg:opt},
we obtain our main result for the computation of hypergraph resolvents.

\begin{theorem}\label{thm:main1} For a connected hypergraph with potential $\cU(\vx)$, 
Algorithm~\ref{alg:opt} with prox-generating function $\nicefrac{1}{2} \cdot \|\cdot\|^2_{\mD}$ computes a multiplicative $\epsilon$-approximation to the optimum of Problem~\ref{eq:resolvent} in $O\left(\frac{1}{(\lambda_G+\lambda) \epsilon^2}\right)$ iterations. Each iteration requires linear time in the size of the hypergraph and a single computation of the discrete-time heat diffusion step of Equation~\ref{eq:0-s-discrete-update-rule}.
\end{theorem}

Our algorithm can be interpreted as a non-linear analogue of  the expansion of Equation~\ref{eq:expansion}. The main step in each iteration are  i) an addition of external charge $\eta \mD^{-1}s$ in Line 5 and ii) a discrete-time heat-diffusion computation in Line 6.  When $\cL$ is linear, it is possible to switch the order of these steps and move all external charge steps to the beginning of the algorithm, yielding an expansion equivalent to that in Equation~\ref{eq:expansion}.

Compared to the graph case, our running time suffers from a polynomial rather than a logarithmic dependence on $\epsilon$, which is a consequence of the non-smoothness of our potential $\cU$. Just as in the case of graphs, the convergence of our algorithm is inversely proportional to $\lambda_G$, slowing down when the spectral gap is small, i.e., when the  heat diffusion mixes slowly. 
It is an interesting open question to obtain hypergraph analogues of the fast Laplacian solver of~\citet{spielmanNearlylinearTimeAlgorithms2004}, which runs in nearly-linear-time in the size of the graph with no dependence on $\lambda_G$ or any other graph parameters. 

In Section~\ref{sec:empirical-comparison}, we evaluate our algorithm against existing methods, with a particular focus on the PPR vector heuristics of~\citet{takai2020hypergraph}. There, we also propose and evaluate our own heuristic algorithm, based on a different choice of $\mR$, which we conjecture removes the poor dependence on $\lambda_G$, at the cost of a slightly worse dependence on the size of the hypergraph.

\section{First-order methods for non-differentiable squared norms}\label{sec:optimization}

\begin{algorithm}[tb]
   \caption{with prox-generating function $\nicefrac{1}{2}\cdot \|\cdot \|_{\mR}^2$ for $\min F(\vx) - \ip{\vs,\vx}$ in Theorem~\ref{thm:optimization}}\label{alg:opt}
\begin{algorithmic}[1]
   \STATE {\bfseries Input:} Oracle access arbitrary $\vz \in \partial F(\vx)$ on input $\vx$; \: $T \in \N$; \: $\epsilon \in (0,1).$
   \STATE $\vx_0 = 0$  \hfill $\triangleright$ Initialization
  \STATE $\eta =  \nicefrac{\epsilon}{2u_{\mR}}$ \hfill $\triangleright$ {Step-size}
   \FOR{$t=0, ... , T-1$}
   \STATE $\hat{\vx}_{t} \gets \vx_t + \eta\mR^{-1} \vs$
   \STATE $\vx_{t+1} \gets \hat{\vx}_{t} - \eta\mR^{-1}\vz_{t}$ with $\vz_t \in \partial F(\hat{\vx}_{t})$
   \ENDFOR
   \STATE $\bar{\vx}_T = \frac{1}{T} \sum_{t=0}^{T-1}\hat{\vx}_t$
   \RETURN $\vx^{\textrm{out}}_T = (1 - \nicefrac{\epsilon}{2}) \cdot \bar{\vx}_T$
\end{algorithmic}
\end{algorithm}

In this section, we present Algorithm~\ref{alg:opt} for the solution of optimization problems of the form
\begin{equation}\label{eq:optimization-objective}
 \textrm{OPT}_{F,\vs} = \min_{\vx \in \R^n} F(\vx) - \ip{\vs,\vx},
\end{equation}
where $\vs \in \R^n$ and $F: \R^n \to \R_{\geq 0}$ has the form $F(\cdot) = \nicefrac{1}{2} \|\cdot\|^2$, for a (potentially non-smooth)  norm $\| \cdot\|.$ 
Notice that $\textrm{OPT}_{F,\vs} = - F^*(\vs)$ is always non-positive, as the Fenchel dual $F^*$ is also a squared norm.
The application to resolvent computation in Problem~\ref{eq:resolvent} will follow from the following lemma, which is proved in Section~\ref{appendix:proof-of-optimization-convergence}.

\begin{lemma}\label{lemma:energy-is-seminorm}
Given a connected hypergraph $G$ and $\lambda \geq 0,$ the function $\cU(\cdot) + \nicefrac{\lambda}{2} \norm{\cdot}_{\mD}^2$ is a squared norm over $\R^V \perp_{\mD} \vone$.
\end{lemma}

Algorithm~\ref{alg:opt} is based on optimistic mirror descent~\cite{rakhlinOnlineLearningPredictable2013} with prox-generating function given by the squared norm $\frac{1}{2}\|\cdot \|^2_{\mR}$ for a positive-definite linear operator $\mR$. It converges to a multiplicative approximation of $\textrm{OPT}_{F,\vs}$ at a rate that only depends on the Poincar\'e constants $\ell_{\mR}, u_{\mR}$, i.e the optimal constants satisfying: 
\begin{equation}\label{eq:poincare}
\forall \, \vx  \in \R^n \setminus \{0\} \: : \:\ell_{\mR} \leq \frac{F(\vx)}{\frac{1}{2}\|\vx\|_{\mR}^2} \leq u_{\mR}.
\end{equation}

The following guarantee is proved in Section~\ref{appendix:proof-of-optimization-convergence} in the Supplementary Material. 
\begin{theorem}\label{thm:optimization}
For any squared norm $F$ satisfying Equation~\ref{eq:poincare} with respect to some $\ell_{\mR},u_{\mR} >0$, error parameter $\epsilon >0 $, and choice of $T = O\left(\frac{u_{\mR}}{\ell_{\mR}\epsilon^2}\right)$, the point $\vx^{\text{out}}_T$ produced by Algorithm~\ref{alg:opt} satisfies
\[
    F(\xout_T) - \ip{\vs,\xout_T} -  \mathrm{OPT}_{F,\vs} \leq \epsilon \cdot |\mathrm{OPT}_{F,\vs}|.
\]
\end{theorem}

To emphasize the novelty of this result, we remark again that, to the best of our knowledge, the condition in Equation~\ref{eq:poincare} cannot be exploited by any existing first-order methods for this problem.
It is only equivalent to $\ell_{\mR}$-strong-convexity and
$u_{\mR}$ strong convexity when $F$ itself arises from a Mahalanobis norm $\|\cdot\|_{\mA}$, which is not the case in our application. 

\section{Experiments}\label{sec:experiments}

\emph{Manifold learning} problems include identification and recovery of some lower-dimensional parametrization of datapoints embedded in some higher-dimensional space. 
Many approaches utilize spectral methods on graphs 
 that are constructed from the geometric embedding of datapoints \cite{belkin2001laplacian, belkin2003laplacian, yan2006graph, talmon2013diffusion}. 
 Many spectral graph methods focus on partitioning datapoints, in either an unsupervised (e.g. using eigenvectors of the Laplacian) or a supervised manner (e.g. propagating a small set of known labels using the graph heat diffusion). In this section, we consider semi-supervised community detection tasks in the presence of manifold structure. In these problems, communities with low-dimensional parametrization are present within the data and a small number of true labels are known. Experiments evaluating our resolvent algorithm of Section~\ref{sec:resolvent} are presented in Section~\ref{sec:empirical-comparison}.
 
 Basic graph-based methods may suffer when data is sampled at heterogeneous densities or when communities only become linearly separable in some augmented feature space. Algorithmic extensions have been proposed to deal with each of these specific instances \cite{li2019survey, you2016community, donoho2003hessian, li2016low}, but in practice correctly employing these extensions requires identifying and appropriately responding to the failure mode present in the revelant dataset. We present numerical experiments demonstrating that, unlike diffusions on graphs, diffusions on hypergraphs can overcome each of these distinct challenges without needing problem-specific algorithmic extensions. Our results are illustrated in Figure~\ref{fig:manifold_learning}. Here we provide an overview, and provide full detail in Section~\ref{ssec:manifold-learning-details}.

We consider three separate problem instances, described in Figure~\ref{fig:manifold_learning}. In each case, data are sampled noisily from each community, and a small number of true labels are revealed. We construct the $k$-nearest neighbor graph and hypergraph. We then consider the discrete-time diffusion resulting from the energy functional induced by taking $\norm{\cdot}_h = \norm{\cdot}_\infty$ (see Equation~\ref{eq:infinity_energy_functional}). These diffusions are initialized using a vector $\vx_0$ constructed from a small number of randomly-sampled true labels. We then run the diffusion for some number of iterations $t$, and use vector $\vx_t$ to estimate true labels.

As illustrated in Figure~\ref{fig:manifold_learning}, diffusions along the nearest-neighbor hypergraph better-capture community structure compared to graph diffusions. The top row of Figure~\ref{fig:manifold_learning} compares area-under-the-ROC-curve (AUC) values achieved by each method. AUC values measure tradeoff between false-positive and true-positive rates, with AUC values close to 1.0 signaling that most sweep-cut vectors identify the communities nearly perfectly. Comparing the AUC distributions highlights the hypergraph's ability to overcome the three very different kinds of challenging structure present in the problem instances. In contrast, the sample estimated labels demonstrate how the generic graph method experiences different failure modes in each case.
\begin{figure*}[t]
\begin{center}
\includegraphics[width=1.0\textwidth]{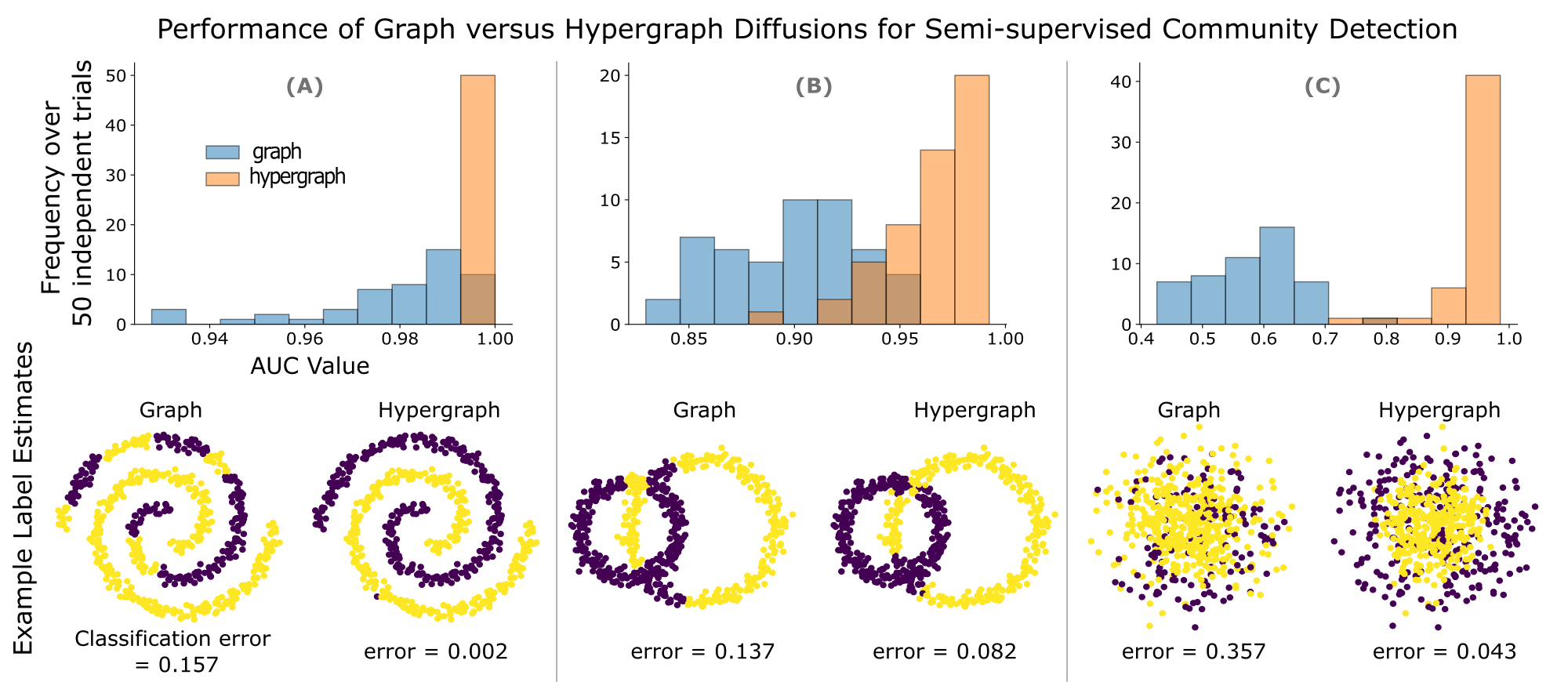}
 \caption{
 Comparing hypergraph versus graph diffusions on semi-supervised community detection problems with manifold structure. We consider three separate problems in which ground-truth communities are (A) interlocking spirals in $\R^2$, (B) overlapping rings of unequal diameter in $\R^2$, and (C) concentric hyperspheres in $\R^5$. Experiment procedure is detailed in Section~\ref{ssec:manifold-learning-details}. 
 }\label{fig:manifold_learning}
\end{center}
\end{figure*}

\section{Acknowledgements}

 LO is supported by NSF CAREER 1943510. AC and AD are supported by NSF DGE 2140001.

\bibliography{lib/bib, orecchia, icml2023_fixed}
\bibliographystyle{plainnat}

\newpage
\appendix
\onecolumn

\section*{Supplementary Material}

\section{Notation and mathematical preliminaries}\label{sec:appendix-preliminaries}
For $T \in \mathbb{N}$, we denote by $[T]$ the set $\{1,2,\dots,T\}$.

\paragraph{Hypergraph preliminaries} A weighted hypergraph $G=(V,E,\vw)$ is a collection of vertices $V$ and hyperedges $E \subseteq 2^V$ with non-negative hyperedge weights $\vw \in \R^{E}_{\geq 0}.$ 
The degree $\deg(i) \defeq \sum_{h \ni i} w_h$ of a vertex $i \in V$ is the sum of the weights of hyperedges containing $i$. Given a set $S\subseteq V$, its \emph{volume} is the size of the degrees of its vertices: $\operatorname{vol}(S) \defeq \sum_{i\in S} \deg(i)$.  We denote by $\mD \in \R^{V\times V}$ the diagonal matrix of degrees of $G$. A hypergraph is said to be \emph{$k$-uniform} if every hyperedge has cardinality $k$. In this sense, a graph is simply a $2$-uniform hypergraph. For any $S\subseteq V$ we denote by $\partial S \subseteq E$ the \textit{boundary} of the set $S$, i.e. $h\in \partial S$ if $\exists i, j\in h$ such that $i\in S, j\in \overline{S}$.

\paragraph{Linear algebra preliminaries} We denote vectors and matrices in boldface e.g. $\vx \in \R^V$ and $\mA \in \R^{n\times n}$. Given a vector $\vx\in \R^V$ the vector $\vxh\in \R^h$ is the restriction of $\vx$ on $h$, i.e. the vector with entries $\vx(i)$ for $i\in h$.  We use $\1$ to denote the all-one vector in $\R^n$. Given two vectors $\vx,\vy \in \R^n$, we denote by $\ip{\vx,\vy}$ the standard inner product between $\vx$ and $\vy$, and we write $\vx \perp \vy$ to indicate $\ip{\vx,\vy}=0$. Similarly, we use $\ip{\vx,\vy}_{\mA}$ and $\vx \perp_{\mA} \vy$ for the same expressions with respect to the inner product given by a positive definite operator $\mA.$ We may also use $\vone_h$ to denote the all-one vector in $\R^h$. We recall that a norm $\norm{\cdot}: \R^n \to \R$ is any function satisfying the following properties:
\begin{description}
    \item[Triangle inequality] $\norm{\vx+ \vy} \leq \norm{\vx} + \norm{\vy}$ for every $\vx,\vy\in \R^n$,
    \item[Absolute homogeneity] $\norm{\alpha \vx} = |\alpha| \norm{\vx}$ for every $\vx \in \R^n$ and $\alpha \in \R$,
    \item[Positive definiteness] $\norm{\vx} \geq 0$ for all $\vx \in \R^n$ and $\norm{\vx}=0 \iff \vx = 0$.
\end{description}
For any $p \in \mathbb{N}$ the $\ell_p$-norm is the norm given by:
\[
    \norm{\vx}_p \defeq \sqrt[p]{\sum_{i\in n} |\vx(i)|^p
    }
\]
and the $\ell_\infty$-norm is given by:
\[
    \norm{\vx}_\infty \defeq \max_{i\in [n]} |\vx(i)|.
\]
Given a norm $\norm{\cdot}$ on $\R^V$, its dual norm is the norm defined as:
\[
    \norm{\vy}_* \defeq \max_{\norm{\vx}\leq 1} \ip{\vx,\vy}.
\]
For any positive-definite matrix $\mA \in \R^{n\times n}$ the norm induced by $\mA$ is $\norm{\vx}_{\mA}\defeq \sqrt{\vx^\top \mA \vx}$. 
\paragraph{Subgradients and convexity} The subdifferential of a convex function $f:\R^V\to \R$ at a point $\vx$ is the set:
\[
    \partial f(\vx) \defeq \{\vy\in \R^n \mid f(\vz) \geq f(\vx) + \ip{\vz-\vx,\vy}, \forall \vz \in \R^n\}.
\]
When $f$ is a convex function, $\partial f(\vx)$ is non-empty at every point $\vx$ in the interior of the domain of $f$. Moreover, if $f$ is convex and differentiable then $\partial f(\vx) = \{\nabla f(\vx)\}$ for every $\vx$ in the interior of the domain of $f$.
An element of the subdifferential of $f$ at $\vx$ is called a \emph{subgradient} of $f$ at $\vx$.
We will make repeated use of the following simple properties:

\begin{property}\label{fact:norm-condition} For every $\vx \in \R^V$:
    $
    \partial \left({1\over 2}\norm{\vx}^2\right) = \norm{\vx} \cdot \partial \left(\norm{\vx}\right)= \norm{\vx}\cdot \argmax_{\|\vy\|_* \leq 1} \ip{\vy,\vx}.
    $
\end{property}
\begin{property}\label{fct:doublenorm}
    If $\vz \in \partial\left({1\over 2} \norm{\vx}^2\right),$ then $\ip{\vz, \vx} = \norm{\vx}^2.$ 
\end{property}
These properties follow from standard results about subgradients we refer the reader to the book of Boyd and Vandenberghe~\cite{boyd2004convex} for more details.

\paragraph{Submodular functions} Given a finite set $h$ a function $\delta : 2^h \to \R$ is submodular if, for all $A, B \subseteq h$ we have:
\[
    \delta(A) + \delta(B) \geq \delta(A\cap B) + \delta (A \cup B).
\]
The Base polyhedron of a submodular function $\delta:2^H \to \R$ is the set:
\[
    B(\delta) \defeq \{\vy \in \R^h \mid \vy(h) = \delta(h), \vy(S) \leq\delta(S) \text{ for every }S \subseteq h\},
\]
where, for any $S \subseteq h$, $\vy(S)\defeq \sum_{i\in S} \vy(i)$. The Lov\'asz extension of $\delta$ is the function $\bar{\delta}:\R^h \to \R$ given by:
\[
    \bar{\delta}(\vx) \defeq \argmax_{\vy\in B(\delta)} \ip{\vy ,\vy}
\]
Submodular functions are a central object of study in combinatorial optimization and have found numerous applications throughout Computer Science. We refer the reader to the monograph by Bach~\cite{bachLearningSubmodularFunctions2013} for more information on subbmodular functions and their applications.

\section{Diffusions and 
energy functionals in the graph setting}\label{sec:graph_case}

In this section, we review the graph equivalent of some of the concepts described in this paper. We do not introduce any novel results here. Rather, we aim to provide guidance for the reader who may not be familiar with known results about graphs.

\paragraph{Graph matrices} Given a weighted, undirected graph $G=(V,E,w)$, with $n$ vertices and $m$ edges, its adjacency matrix is the matrix $\mA\in \R^{n\times n}$ given by:
\[
    \mA_{i,j} = \begin{cases}
        w_{ij}& \text{if }i \sim j,\\
        0 & \text{otherwise,}
    \end{cases}
\]
its degree matrix is the diagonal matrix given by $\mD_{i,i} = \deg(i)$, where $\deg(i)$ is the weighted degree of vertex $i$ in $G$. The weight matrix is the diagonal matrix $\mW \in \R^{m\times m}$, where $\mW_{e,e} = w_e$. The Laplacian matrix is the matrix $\mL \in \R^{n\times n}$ defined as $\mL \defeq \mD-\mA$. Even though the graph considered is undirected, we will assume that we have fixed a choice of edge orientation for each edge. The incidence matrix of $G$ is the matrix $\mB \in \R^{m\times n}$ having one row $\boldsymbol{\chi}_{uv}\in \R^n$ for each edge $uv$ given with:
\[
    \boldsymbol{\chi}_{uv}(i) = \begin{cases}
        -1 &\text{if }i=u,\\
        1 &\text{if }i=v,\\
        0 &\text{otherwise.}\\
    \end{cases}
\]

It is easy to see that, regardless of the choice of edge orientation, we have: $\mL = \mB^{\top}\mW\mB$.

\paragraph{Spectral gap} Given the matrices defined above, the \emph{spectral gap} of $G$ is the quantity:
\[
    \lambda (G) \defeq \min_{\vx \perp_{\mD} \one} {\vx^T \mL \vx \over \vx^T \mD\vx},
\]

i.e. the second smallest generalized eigenvalue of $\mL$ with respect to the $\norm{.}_\mD$ norm. A cornerstone of spectral graph theory is Cheeger's inequality which relates the value of $\lambda(G)$ to the conductance of $G$:
\begin{equation}\label{eq:cheeger}
    {\lambda(G) \over 2} \leq \phi(G) \leq \sqrt{2\lambda(G)}.
\end{equation}
The proof of the above result also yields fast algorithms to finding a small conductance cut of $G$.

\paragraph{Heat diffusion} A prototypical example of diffusions on graph is given by the heat equation:
\begin{equation}\label{eq:heat-equation}%
    \begin{cases} 
        \dot{\vx}(t) = -\mD^{-1}\mL \vx(t)\\
        \vx(0) = \vx_0,
    \end{cases}
\end{equation}

Equation \eqref{eq:heat-equation} is a discrete-space analogue of the heat equation on manifolds. It describes a system in which the temperature of each vertex moves towards the average temperature of its neighbors. It is a linear system of ordinary differential equations and its solution is the \emph{heat kernel}:
\begin{equation}
    \vx(t) = e^{-t\mD\mL}\vx(0).
\end{equation}

From the above, it follows that the worst-case (over the choice of $\vx_0 \perp 1$) convergence of \eqref{eq:heat-equation} is directly linked to  the spectral gap $\lambda_2(\mD^{-1}\mL) = \lambda_2(G)$.

Diffusions often arise as steepest descent flows that lead to the minimizer of some objective function. For instance, Equation~$\eqref{eq:heat-equation}$ leads to minimization of the Laplacian energy functional:
\[
    \cU_G(\vx) = {1\over 2}\vx^\top L\vx,
\]
and it follows the steepest descent flow with respect to the $\norm{\cdot}_D$-norm. 

\paragraph{Random walks} 
A fundamental stochastic process on a graph $G$ is its natural random walk, in which one starts at a random vertex $v$, and then proceeds to repeatedly move to a neighbor $u$ chosen at random with probability proportional to the weight $w_{uv}$ of the edge between $u$ and $v$ in $G$. It is easy to see that if $v$ is selected according to a probability distribution $\vp_{t}$ then the probability distribution of the next vertex $u$ is given by:
\begin{equation}
    \vp_{t+1} = \mA\mD^{-1} \vp_{t}.
\end{equation}

In particular, the above yields that if the first vertex of a random walk is chosen according to a probability distribution $\vx_0$ after $t$ steps the walker finds themselves at a random vertex chosen according to the probability distribution: $\vp_{t} = (\mA\mD^{-1})^t \vp_{0}$. This process is a reversibile Markov chain which may or may not be periodic. In order to ensure periodicity and guarantee convergence, one often focuses on the \emph{lazy random walk process} instead. Here, at each iteration the walker flips a fair coin, and with probability $1/2$ they remain where they are, while with the remaining probability they perform a step of the standard random walk. The evolution of the probability distribution under this process is given by:
\begin{equation}
    \vp_{t+1} = {1\over 2} (\mI + \mA\mD^{-1})\vp_{t}.
\end{equation}

Rewriting the above we obtain:
\begin{equation}\label{eq:lazy-rw-alg}
    \vp_{t+1} = \vp_t - {1\over 2} \mL \mD^{-1}\vp_t.
\end{equation}

Moreover, by defining $\vx_t = \mD^{-1}\vp_t$ Equation~\eqref{eq:lazy-rw-alg} gives:
\[
    \vx_{t+1} = \vx_t - {1\over 2} \mD^{-1} \mL \vx_t.
\]
which is the forward Euler discretization of the dynamic in \eqref{eq:heat-equation}. In particular, this shows that, up to a change of basis, the lazy random walk process can be interpreted as a discrete-time version of the heat diffusion process. 

\paragraph{Personalized PageRank vectors} Another central object of study is the personalized PageRank (PPR) vector. Here, one is given a starting seed distribution $\vs\in \R^V$ over the vertices of a graph $G$ as well as teleportation constant $\alpha\in (0,1)$. The PageRank process amounts to selecting a geometrically distributed number of steps $T \sim Geom(\alpha)$ and then taking $T$ steps of the lazy random walk process starting from a random vertex sampled according to $\vs$. The resulting probability vector $\vp_\alpha(\vs)$ is the personalized PageRank (PPR) vector given by:
\begin{equation}\label{eq:}
    \vp_\alpha(s) = {\alpha} \left(\sum_{t=0}^\infty (1-\alpha)^t {1\over 2^t} (I+\mA\mD^{-1})^t \right) \vs.
\end{equation}

The PPR vector $\vp_\alpha(s)$ is also the unique fixed point of the iterative process given by:
\begin{equation}\label{eq:graph-pagerank-fixed-point}
    \vp_{t+1} = \alpha \vs + (1-\alpha) {1\over 2}\left(\mI+\mA\mD^{-1}\right)\vp_t,
\end{equation}
this allows one to interpret the PPR as the unique stationary distribution of a stochastic process in which one starts at an arbitrary vertex of a graph and repeatedly does the following: with probability $\alpha$ they move to a new random vertex selected according to $\vs$ and with probability $1-\alpha$ they perform a step of the lazy random walk.

While these probabilistic interpretations help gain intuition on the behavior of the PPR vector, in the interest of analysis and generalization we prefer to adopt a variational perspective on it. To this end, we note that the PPR vector is closely linked to the solution to the following optimization problem:
\begin{equation}\label{eq:graph-variational-pagerank}
    \underset{\vx \in \R^n}{\operatorname{minimize}} \quad {a\over 2}\vx ^T \mL \vx +{b\over 2}\norm{\vx - \vx_\vs}^2_{\mD}.
\end{equation}

In particular, setting $a={1\over 2}$ and $b={2\alpha-1 / 1-\alpha}$ and $\vx_\vs = \alpha / (2\alpha-1)\mD^{-1} \vs$, then the optimal solution $\vx^*$ to \eqref{eq:graph-variational-pagerank} equals $\mD^{-1} \vp_{\alpha}(\vs)$. This correspondence between the probabilistic and the variational definitions of PageRank is generalized for the case of hypergraphs in Lemma~\ref{lem:hypergraph-pagerank-fixed-point}.

\section{Hypergraph heat diffusions}\label{sec:properties}

In this section we show properties of the subdifferential $\cL(\vx)$, properties of the trajectories $\vx(t)$ that arise from the ODE in Theorem~\ref{thm:characterizing_subgradient_flow}, and properties of the discretized sequence $\{\vx_t\}_{t=0}^T$ as in Theorem~\ref{thm:discrete-time-diffusion}. We begin with two straightforward consequences of the form of the Laplacian operator $\cL$ given in Equation~\ref{eq:laplacian}.
\begin{property} \label{prop:fixed-point}
For any $\alpha \in \R$, $\cL(\alpha \vone) = \vec{0}.$  
\end{property}
\begin{proof}
    For any norm $\norm{.}_h$,
    $
        \min_{u\in \R}\norm{\alpha\vone_h - u\1_h}_h = 0
    $
    The result thus follows from Equation~\ref{eq:laplacian}:
    \[
        \cL(\alpha\vone) = \left\{\sum_{h\in E} w_h \left(\min_{u\in \R}\norm{\alpha\vone_h - u\1_h}_h\right) \vy_h) \right\} = \vec{0}
    \]
\end{proof}
\begin{property} \label{prop:inner-prod-with-cL}
    For any $\vx\in \R^n$,
    \[
        \langle \vx, \vz\rangle = 2\cU(\vx) \quad \forall \vz\in\cL(\vx).
    \]
\end{property}
\begin{proof}
    By Equation~\ref{eq:laplacian}, $\forall \vz \in \cL(\vx)$
    \[
        \vz = \sum_{h\in E} w_h \left(\min_{u\in \R}\norm{\vxh - u\1_h}_h\right) \vy_h \quad \text{ where} \quad \vy_h \in \argmax_{\substack{\vy \perp \vone\\ \norm{\vy}_{h,*} \leq 1}} \ip{\vy, \vx_h}
    \]
    Given this characterization, for each $\vy_h$
    \[
        \langle \vx_h,\vy_h\rangle = \max_{\substack{\vy \perp \vone\\ \norm{\vy}_{h,*} \leq 1}} \langle \vx_h, \vy\rangle = \max_{\norm{\vy}_{h,*} \leq 1}\min_{u\in\R} \langle \vx_h - u\vone, \vy\rangle = \min_{u\in\R}\norm{\vx_h - u\vone}_h
    \]
    Thus
    \[
        \langle \vx, \vz\rangle = \sum_{h\in E} w_h \left(\min_{u\in \R}\norm{\vxh - u\1_h}_h\right) \langle \vx_h,\vy_h\rangle = \sum_{h\in E} w_h \left(\min_{u\in \R}\norm{\vxh - u\1_h}_h\right)^2 = 2\cU(\vx)
    \]
    by the definition of $\cU(\vx)$ in Equation~\ref{eq:our-setting}.
\end{proof}

We now consider functions $\vx(t)$, solutions to an ODE of the form in Theorem~\ref{thm:characterizing_subgradient_flow}, and establish some fundamental properties of this trajectory.
\begin{property}\label{prop:mean-preservation}
For $\vx(t)$ as in Theorem~\ref{thm:characterizing_subgradient_flow},  $\nicefrac{d}{dt}( \ip{\vone,\vx(t)}_{\mD}) = 0$ almost everywhere on  $t \in [0,\infty)$.
\end{property}
\begin{proof}
    Wherever $\vx(t)$ is differentiable, $ \nicefrac{d}{dt}\ip{\vone,\vx(t)}_{\mD} = \ip{\vone,\dot{\vx}(t)}_{\mD}$. By Theorem~\ref{thm:characterizing_subgradient_flow}, $\dot{\vx}(t) \in -\mD^{-1}\cL(\vx(t))$ a.e. on $t\in [0,\infty)$. Everywhere this holds, we thus have
    \[
        \nicefrac{d}{dt}\ip{\vone,\vx(t)}_{\mD} = -\ip{\vone,\mD^{-1}\cL(\vx(t))}_{\mD} = -\ip{\vone,\cL(\vx(t))} = 0
    \]
    as, by Equation~\ref{eq:laplacian}, $\cL(\vx) \perp \vone$ for all $\vx \in \R^V$.
\end{proof}
This property has the following immediate consequence:
\begin{property}\label{property:pi-unchanging-cts-time}
    Recall the definition $\pi(\vx) \defeq \nicefrac{\ip{\vx,\vone}_{\mD}}{\|\vone\|^2_{\mD}} \cdot \vone$. Then, for any $\vx(t)$ as in Theorem~\ref{thm:characterizing_subgradient_flow}, $\pi(\vx(t)) = \pi(\vx_0) \ \forall t\in [0,\infty)$.
\end{property}

\paragraph{Maximum principle} The standard theory of diffusion processes~\cite{} typically requires that a diffusion also obey a maximum principle, which asks that $\cL_i(\vx)$ be non-negative whenever vertex $i$ is a local maximum of $\vx$ , i.e., all vertices $j$ that share a hyperedge with $i$ must have $\vx_j \leq \vx_i.$
Equivalently, this asks that the diffusion process infinitesimally decreases the $\vx$-value at local maximimum of $\vx.$ The hypergraph heat diffusion satisfies this property under some further assumptions on the norms $\|\cdot \|_h$. Namely, we'll show that a maximum principle holds whenever the norms $\norm{\cdot}_h$ are \textit{monotonic}:

\begin{definition}
    A norm $\norm{\cdot}$ is \textit{monotonic} if for any $\vx, \vy\in \R^n$, $\abs{\vx_i} \leq \abs{\vy}_i \ \forall i\in [n]$ implies $\norm{\vx}\leq \norm{\vy}$
\end{definition}

Subgradients of monotone norms have the following property, which will be key in establishing our maximum principle.
\begin{property}\label{prop:subgrad-of-monotone-norms}
    Given $\norm{\cdot}$ monotonic and $\vx$ such that $\max_i \vx_i \geq 0$ and $\min_i \vx_i \leq 0$, then for all $\vy\in\partial\norm{\vx}$,
    \[
        \vy_{i_\text{max}} \geq 0 \quad \text{ for }\quad i_\text{max} \in \argmax_{i\in [n]} (\vx_i) \qquad\text{and}\qquad \vy_{i_\text{min}} \leq 0 \quad \text{ for }\quad i_\text{min} \in \argmin_{i\in [n]} (\vx_i)
    \]
\end{property}
\begin{proof}
    We begin by considering $i_\text{max}$. Consider a vector $\vz$ such that $\vz_i = \vx_i$ for all $i\neq i_\text{max}$, and
    \[
        \vz_i \leq \vx_i \quad \text{ and }\quad |\vz_i| \leq |\vx_i| \quad \text{ for } i=i_\text{max}.
    \]
    Observe that it is always possible to construct such a $\vz$ given $\max_i \vx_i \geq 0$. By the definition of the subgradient, for all $\vy\in\partial\norm{\vx}$,
    \[
        \norm{\vz} \geq \norm{\vx} + \langle \vz-\vx,\vy\rangle
    \]
    Because $\norm{\cdot}$ is monotone, and since by construction $|\vz_i| \leq |\vx_i|$ for all $i$, $\norm{\vz} \leq \norm{\vx}$. Thus for all $\vy\in\partial\norm{\vx}$ the above implies
    \[
        0 \geq \langle \vz-\vx,\vy\rangle = (\vz_{i_\text{max}}-\vx_{i_\text{max}})\vy_{i_\text{max}}
    \]
    Because, by construction, $\vz_{i_\text{max}}-\vx_{i_\text{max}}\leq 0$, we conclude $\vy_{i_\text{max}}\geq 0$.
    
    An analogous argument yields the result for $\vy_{i_\text{min}}$.
\end{proof}

\begin{lemma}\label{lemma:maximum_prinicple}
    Consider a hypergraph potential $\cU(\vx)$ generated by hyperedge norms $\norm{\cdot}_h$ such that $\forall h\in E$, $\norm{\cdot}_h$ is monotone. Then
    \[
        \cL(\vx)_{i_\text{max}} \geq 0 \quad \text{ for }\quad i_\text{max} \in \argmax_{i\in [n]} (\vx_i) \qquad\text{and}\qquad \cL(\vx)_{i_\text{min}} \leq 0 \quad \text{ for }\quad i_\text{min} \in \argmin_{i\in [n]} (\vx_i)
    \]
\end{lemma}
\begin{proof}
    We begin by considering $i_{\text{max}}$. Observe that for $\norm{\cdot}_h$ monotone,
    \[
        \argmin_{u\in \R}\norm{\vxh - u\1_h}_h \in [\min_{i}\vxh(i), \max_{i}\vxh(i)]
    \]
    i.e. for any $h$, the minimizing shift of $\vone$ sits between the extreme values of the restriction to $\vx$ on $h$. Thus for all $h$, Property~\ref{prop:subgrad-of-monotone-norms} implies that for all $\vy_h \in \partial(\min_{u \in \R} \norm{\vxh - u\1_h}_h)$,
    \[
        \vy_h(i) \geq 0 \quad \text{ for } i \in \argmax_{j} \vxh(j).
    \]
    Consider $i_{\text{max}}$ an index of a global maximum entry of $\vx$. For all $h\in E$, we have two options: either $i_{\text{max}} \not\in h$ in which case $\vy_h(i_{\text{max}}) = 0$ for all $\vy_h \in \partial(\min_{u \in \R} \norm{\vxh - u\1_h}_h)$, or $i_{\text{max}} \in \argmax_{j} \vxh(j)$ in which case $\vy_h(i_{\text{max}}) \geq 0$. Thus, by the definition of $\cL(\vx)$,
    \[
        \cL(\vx) = \partial\cU(\vx) = \left\{\sum_{h\in E} w_h \left(\min_{u\in \R}\norm{\vxh - u\1_h}_h\right) \vy_h) \right\},
        \textrm{ where } 
        \vy_h \in \partial(\min_{u \in \R} \norm{\vxh - u\1_h}_h)
    \]
    we conclude that for $i_{\text{max}}$ an index of a global maximum entry of $\vx$, $\cL(\vx)_{i_{\text{max}}} \geq 0$.
    
    The result for $i_{\text{min}}$ follows from an analogous argument.
\end{proof}

\section{Applications of hypergraph heat diffusions: local hypergraph partitioning}\label{sec:local}
Our discrete-time heat diffusion can be directly plugged into the existing graph partitioning algorithm of~\citet{ikedaFindingCheegerCuts2019}. Their algorithm, which is a hypergraph analogue of the original local graph partitioning algorithm of Spielman and Teng~\cite{spielman2013local}, is based on approximating the continuous-time heat diffusion over the hypergraph, but does not provide a direct way of numerically performing this task. We show that the same analysis carries through when the continuous-time diffusion is replaced by our discrete-time diffusion, establishing Cheeger-like approximation results with a running time that is linear in the size of the hypergraph and inversely proportional to the spectral gap. 

We need to remark at this point that the analysis of~\citet{ikedaFindingCheegerCuts2019} relies on a crucial unproven structural assumption on the continuous-time heat diffusion for a connected hypergraph $G$: for a uniformly random starting vector $\vx(0) \perp_{\mD} \vone$ with $\|vx(0)\|_{\mD}^2 =1$, we have:
$$
\frac{\cU(\vx(t))}{\nicefrac{1}{2} \norm{\vx(t)}_{\mD}^2} \rightarrow \lambda_G \; \textrm{ as } \; t \rightarrow \infty.
$$
This is easy to show in the graph context by considering the eigenvector decomposition of the quadratic form of $\cU(x) = \vx^T \mL \vx$.
However, it is our opinion that this assumption is likely to be false in the hypergraph case for two reasons. First, under this assumption, \citet{ikedaFindingCheegerCuts2019} show approximation results for minimizing hypergraph conductance that break long-standing hardness assumptions. In particular, their algorithm can be directly applied to the vertex-conductance problem to obtain a Cheeger-like approximation that breaks the small-set-expansion conjecture, as shown by~\citet{louis2013complexity}. Secondly, we do not expect the analogue of the eigenvalue problem for non-linear hypergraph Laplacians to exhibit the same benign non-convexity of the linear case. For these reasons, in the following analysis, we adopt a more reasonable assumption, which is a close analogue of a similar assumption taken by~\cite{liu2021strongly} in the context of local hypergraph partitioning by PageRank.

\subsection{Local Hypergraph Partitioning Algorithm}

Define the hyperedge conductance objective $\phi(S)$ for a cut $(S, V\setminus S)$ analogously to graphs as: 
$$
\phi(S) = \frac{\sum_{h \in E} w_h \delta_h^{\infty}(S \cap h)}{\min \{\Vol(S), \Vol(\bar{S})\}} = \frac{w(\partial S)}{\min \{\Vol(S), \Vol(\bar{S})\}}.
$$
where $\partial S$ refers to the boundary of the set $S$ i.e. $\partial S \defeq \{h \in E \mid h\cap S \neq \emptyset \text{ and } h \cap \bar{S} \neq \emptyset\}$ and $w(\partial S)\defeq \sum_{h\in \partial S} w_h$. In the local hypergraph partitioning algorithm, we are asked to approximate the hypergraph conductance of an unknown target set $S$ given a uniformly random seed vertex $v$ within $S.$ Our proposed algorithm~\ref{alg:local-partition} achieves this task by running the discrete-time heat diffusion $\vx_t$ from $v$ for $T$ iterations, choosing time $t^*$ where the distribution $\vx_t$ has minimum Rayleigh quotient and performing a sweep cut of $\vx_{t^*}.$

\begin{algorithm}[h]   \caption{\texttt{LocalPartition}$(G,v,\phi)$}\label{alg:local-partition}
\begin{algorithmic}[1]
   \STATE {\bfseries Input:} A hypergraph $G=(V,E,\vw)$, a starting vertex $v\in V$, target conductance $\phi \in [0,1]$.
   \STATE $\vx_0 = \vone_v$
  \STATE $T \defeq 1/(3\phi)$
   \FOR{$t=1, ... , T$}
   \STATE $\vx_t = \vx_{t-1} - \mD^{-1} \cL(\vx_{t-1})$
   \ENDFOR
   \STATE $t^* =\argmin_{t\in [T]} {\cU(\vx_{t^*})\over \norm{\vx_{t^*} - \pi(\vone_S)}_{\mD}^2}$
   
   \STATE Round the vector $\vx_{t^*}$ to a cut $(U,V\setminus U)$ such that:
   \[
        \phi(U) \leq O\left(\sqrt{\cU(\vx_{t^*})\over \norm{\vx_{t^*} - \pi(\vone_S)}_{\mD}^2}\right)
   \]
   using the hypergraph Cheeger Inequality~\cite{chan2018spectral,yoshida2019cheeger}.
   \RETURN $(U,\bar{U})$
\end{algorithmic}
\end{algorithm}

In the following theorem, we show that under an assumption about the existence of a specific probability distribution for sampling the starting vertex $v$ from the target set $S$, the algorithm \texttt{LocalPartition} can be used to obtain a Cheeger-like approximation to the conductance of $S.$
Our assumption is the heat-diffusion analogue of a closely related assumption made for PPR vectors by~\citet{liu2021strongly} in their Lemma B.3.
It is an interesting open question to prove this assumption true for specific values of the parameter $c_{G,S}$. As discussed above, based on the work of~\citet{louis2013complexity}, it is unlikely that $c_{G,S} = O(1)$ in general.

\begin{theorem}
Let $S\subset V$ be a subset of $V$ with $\Vol(S)\leq \Vol(G)/2$. Let $\{\vx_t^{(S)}\}_{t\in [T]}$ be the sequence of iterated generated by running the heat diffusion with starting vector $\vx^{(S)}_0= \vone_S$ and for any $v \in S$ let $\{\vx^{(v)}_t\}_{t\in [T]}$ be the sequence of iterates generated by running the heat diffusion with starting vector $\vx_0 = \vone_v$. Suppose that there exists a distribution $\cD_T$ such that:
\begin{equation}\label{eq:probabilistic-assumption}
    \underset{v\sim \cD}{\mathbb{E}}\left[ {\cU(\vx_t^{(v)})\over \norm{\vx_t^{(v)} - \pi(\vone_v)}_{\mD}^2}\right] \leq c_{G,S} \, \cdot \, {\cU(\vx_t^{(S)})\over \norm{\vx_t^{(S)} - \pi(\vone_S)}_{\mD}^2}
\end{equation}

for some $c_{G,S} \in \R$, for all $t \in [T]$.
Then running Algorithm~\ref{alg:local-partition} with input $v \sim \cD$ and $\phi \geq \phi(S)$ returns a cut $(U, \bar{U})$ such that:
\[
    \underset{v\sim \cD}{\mathbb{E}}\left[\phi(U)\right] \leq O(\sqrt{c_{G,S} \cdot \phi(S)}),
\]
In particular, by Markov's inequality, the algorithm returns a cut $U$ with $\phi(U)\leq O(\sqrt{c_{G,S} \cdot \phi(S)})$ with constant probability.
\end{theorem}
\begin{proof}
    For $t \leq 1/2\phi(S)$, from Lemma~\ref{lem:local-partitioning-upperbound} and Lemma~\ref{lem:local-partitioning-lower-bound}, we have:
    \[
        \left(\frac{1}{2} -t\phi(S)\right)^2  \leq \left(1 - \min_{s\in [t]}{\cU(\vx^{(S)}_t) \over \norm{\vx^{(S)}_t - \pi(\vx_0)}_{\mD}^2}\right)^t \leq e^{-tR(t)}
    \]
    where
    \[
        R_S(t)\defeq \min_{s\in [t]}{\cU(\vx^{(S)}_t) \over \norm{\vx_t - \pi(\vx_0)}_{\mD}^2}
    \]
    Taking the natural log of both sides yields
    \begin{equation}\label{eq:R(t)_upperbound}
        -tR_S(t) \geq 2\ln\left(\frac{1}{2}-t\phi(S)\right)
    \end{equation}
    For choice of $T$ as in Algorithm~\ref{alg:local-partition}, for some $t\in [T]$ we have $t = \Omega(\phi(S))$ and $t < 1/2\phi(S)$.  For such $t$, Equation~\ref{eq:R(t)_upperbound} implies $R_S(t) \leq O(\phi(S))$. In particular, observe that by definition, $R_S(T)\leq R_S(t)$ for all $T\geq t$. Thus $R_S(T) \leq O(\phi(S))$.
    
    Let $U_{t^*}$ be the (random) set returned by the algorithm when the algorithm is run with a starting vertex $v$ sampled according to $\cD$, we then have:
    \begin{align*}
        \E{v\sim \cD}{\phi(U_{t^*})}^2 &\leq \E{v\sim \cD}{\phi(U_{t^*})^2}\\
        &\leq \E{v\sim \cD}{O\left( { \cU(\vx^{(v)}_{t^*})\over \norm{\vx^{(v)}_{t^*} -\pi(\vone_v)}_{\mD}^2}\right)}\\
        &= \E{v\sim \cD}{\min_{t\in T} O\left( { \cU(\vx^{(v)}_{t})\over \norm{\vx^{(v)}_{t} -\pi(\vone_v)}_{\mD}^2}\right)}\\
        &\leq \min_{t\in T}\E{v\sim \cD}{ O\left( { \cU(\vx^{(v)}_{t})\over \norm{\vx^{(v)}_{t} -\pi(\vone_v)}_{\mD}^2}\right)}\\
        &\leq \min_{t\in T}O\left(c_{G,S}\cdot  { \cU(\vx^{(S)}_{t})\over \norm{\vx^{(S)}_{t} -\pi(\vone_v)}_{\mD}^2}\right) = O(c_{G,S} \cdot R_S(T))\\
        &\leq O(c_{G,S} \cdot \phi(S)).
    \end{align*}
    giving:
    \[
        \E{v\sim \cD}{\phi(U_{t^*})} \leq O(\sqrt{c_{G,S} \cdot \phi(S)})
    \]
    as needed.
\end{proof}

\subsection{Necessary Lemmata}

\begin{lemma}
Let $\vone_S$ be the zero-one indicator of a set $S \subseteq V$. We then have:
\begin{equation}
    \norm{\vone_S - \pi(\vone_S)}_\mD^2 = \ip{\vone_S , \vone_S-\pi(\vone_S)}_{\mD} = {\Vol(S)\Vol(\bar{S}) \over \Vol(G)}
\end{equation}
and:
\begin{equation}
    \phi(S) \leq {\cU(\vone_S) \over \norm{\vone_S - \pi(\vone_S)}_{\mD}^2} \leq 2\cdot \phi(S).
\end{equation}
\end{lemma}
\begin{proof}
We have:
\begin{align*}
    \norm{\vone_S - \pi(\vone_S)}_\mD^2 &= \norm{\vone_S}_{\mD}^2 - 2 \ip{\vone_S ,\pi(\vone_S)}_{\mD} + \norm{\pi(\vone_S)}^2_\mD = \Vol(S) - {\Vol(S)^2 \over \Vol(G)} = {\Vol(S)\Vol(\bar{S}) \over \Vol(G)}.
\end{align*}
A similar calculation shows:
\[
    \ip{\vone_S ,\vone_S-\pi(\vone_S)}_{\mD}^2= {\Vol(S)\Vol(\bar{S}) \over \Vol(G)}.
\]
Then, since $\Vol(S) + \Vol(\bar{S})= \Vol(G)$:
\[
    {\min\{\Vol(S), \Vol(\bar{S})\} \over 2} \leq {\Vol(S)\Vol(\bar{S}) \over \Vol(G)} \leq \min\{\Vol(S), \Vol(\bar{S})\},
\]

from which it follows that:
\[
    \phi(S) \leq {\cU(\vone_S) \over \norm{\vone_S - \pi(\vone_S)}_{\mD}^2} \leq 2\cdot \phi(S).
\]
\end{proof}

\begin{lemma}\label{lem:local-partitioning-norm-shrinkage}
Let $\vx \in \R^V$, if $\cL(\vx)\defeq \partial\cU^{\infty}(x)$ then:
\[
    \norm{\vx - \mD^{-1} \cL^\mD(\vx)}_\infty\leq \norm{\vx}_\infty
\]
i.e. applying a step of the discrete-time heat diffusion with respect to the $\cU^{\infty}$ potential to some vector $\vx$ cannot increase the value of its $\ell_\infty$ norm.
\end{lemma}
\begin{proof}
    Consider the set $\mathcal{B}(i,1) \defeq \{j\sim i\}\cup \{i\}$. Observe that for $\cL(\vx) = \partial\cU^{\infty}(x)$,
    \begin{align*}
        \left(\mD^{-1}\cLD\right)_i &= \sum_{\substack{h\ni i \\ i\in \argmax\{\vxh(i)\}}} \frac{\gamma_h w_h}{d_i}\max_{j\in h}(\vx_i-\vx_j) + \sum_{\substack{h\ni i \\ i\in \argmin\{\vxh(i)\}}} \frac{\gamma_h w_h}{d_i}\min_{j\in h}(\vx_i-\vx_j)
    \end{align*}
    for values $\gamma_h\in[0,1] \ \forall h$.

    Let $S$ denote the set of indices
    \begin{align*}
        S \defeq &\{j \mid i,j\in h, i\in \argmax\{\vxh(i)\}, j \in \argmax_{j\in h}(\vx_i-\vx_j)\}\\
        &\qquad \bigcup \{j \mid i,j\in h, i\in \argmin\{\vxh(i)\}, j \in \argmin_{j\in h}(\vx_i-\vx_j)\}
    \end{align*}
    Then we can express $\left(\mD^{-1}\cLD\right)_i$ as a sum over $S$ with new non-negative weights $w_j$
    \[
        \left(\mD^{-1}\cLD\right)_i = \sum_{j\in S}w_j (\vx_i - \vx_j), 
    \]
    where in particular, given $\gamma_h \in [0,1] \ \forall h$ and the definition of $d_i$, $\sum_{j\in S}w_j \leq 1$. Moreover, observe that $i \in \mathcal{B}(i,1)$, so we can equivalently express
    \[
        \left(\mD^{-1}\cLD\right)_i = \sum_{j\in S}w_j (\vx_i - \vx_j) + \left(1-\sum_{j\in S}w_j\right)(\vx_i-\vx_i)
    \]
    thus establishing
    \[
        \mD^{-1}\cLD(\vx) \in \text{convhull}\{\vx(i)-\vx(j) \mid j \in \mathcal{B}(i,1)\}.
    \]
    In particular, this implies that for $\vx'\defeq \vx - \mD^{-1}\cLD(\vx)$, $\forall i\in V$,
    \[
        \vx'(i) \in \left[\min_{j\in V} \vx'(j),\max_{j\in V}\vx'(j)\right],
    \]
    yielding the desired result.
\end{proof}

\begin{lemma}\label{lem:local-partitioning-upperbound}
Consider the sequence of iterates $\{\vx_t\}_{t\geq 0}$ obtained by running the discrete diffusion as in Theorem~\ref{thm:discrete-time-diffusion}. Then $\forall T$ we have:
\[
    \norm{\vx_T - \pi(\vx_0)}_{\mD}^2 \leq \left(1 - \min_{t\in [T]}{\cU(\vx_t) \over \norm{\vx_t - \pi(\vx_0)}_{\mD}^2}\right)^T\norm{\vx_0 - \pi(\vx_0)}_{\mD}^2.
\]
\end{lemma}
\begin{proof}
    This follows directly from Theorem~\ref{thm:discrete-time-diffusion}.
\end{proof}

\begin{lemma}\label{lem:local-partitioning-lower-bound}
Consider the sequence of iterates obtained by running the heat diffusion with starting vector $\vx_0 = \vone_S$ for some $S \subseteq V$, with $\Vol(S) \leq \Vol(G)/2$. Then, for any $T$ such that $T\phi(S)\leq 1/2$ we:
 \[
    {\norm{\vx_T - \pi(\vone_S)}^2_{\mD}\over \norm{\vone_S - \pi(\vone_S)}^2_{\mD}} \geq \left({1\over 2} - T \phi(S)\right)^2.
\]
\end{lemma}
\begin{proof}
We have:
\begin{align*}
                \norm{\vx_T - \pi(\vone_S)}^2_{\mD} &\geq \norm{\vx_T - \pi(\vone_S)\big|_S}^2_{\mD}\geq \frac{1}{\Vol(S)}\left(\sum_{i\in S}d_i(\vx_T(i) - \pi(\vone_S)_i)\right)^2\\
                &= \frac{1}{\Vol(S)}\left(\vone_S^T \mD(\vx_T-\pi(\vone_S))\right)^2
\end{align*}
where the second inequality follows by the Cauchy-Schwarz inequality. We now seek to lowerbound $\vone_S^T \mD(\vx_T-\pi(\vone_S))$. To do so, we will upperbound the change in this quantity over one iteration using the definition of the update step:
\begin{align*}
    \vone_S^T \mD(\vx_{T-1}-\pi(\vone_S)) - \vone_S^T   \mD(\vx_T-\pi(\vone_S)) &= \vone_S^T \mD(\vx_{T-1}-\vx_{T}) =  \vone_S^T\cL(\vx_{T-1}),
\end{align*}
where, for any vector $\vx$,
\begin{align*}
    \vone_S^T \cL(\vx)
    &=\sum_{i\in S}\sum_{h\ni i} w_h \min_{u\in\R}\norm{\vxh-u\vone}_\infty\left(\partial \min_{u\in\R}\norm{\vxh-u\vone}_\infty\right)_i\\
    &=\sum_{h \subset S} w_h \min_{u\in\R}\norm{\vxh-u\vone}_\infty\langle\partial \min_{u\in\R}\norm{\vxh-u\vone}_\infty, \vone_h\rangle \\
    &\qquad + \sum_{h \in \partial S} w_h \min_{u\in\R}\norm{\vxh-u\vone}_\infty\sum_{i\in h\cap S}\left(\partial \min_{u\in\R}\norm{\vxh-u\vone}_\infty\right)_i
\end{align*}

Recall (as in Equation~\ref{eq:laplacian}) that for any norm $\norm{\cdot}_h$, $\langle\partial \min_{u\in\R}\norm{\vxh-u\vone}_h, \vone_h\rangle = 0$. Moreover, for $\norm{\cdot}_\infty$, 
\[
    \sum_{h \in \partial S}\left(\partial \min_{u\in\R}\norm{\vxh-u\vone}_\infty\right)_i \leq 1.
\]
Thus we have
\begin{align*}
    \vone_S^T \cL(\vx) &\leq \sum_{h \cap S \neq \emptyset} w_h \min_{u\in\R}\norm{\vxh-u\vone}_\infty \leq \norm{\vx}_\infty w(\partial S)
\end{align*}
In particular, for $\vx = \vx_{T-1}$ by \ref{lem:local-partitioning-norm-shrinkage}, 
\[
    \norm{\vx_T}_\infty \leq \norm{\vx_0}_\infty = \norm{\vone_S}_\infty= 1.
\]
Thus
\[
    \vone_S^T \mD(\vx_{T-1}-\pi(\vone_S)) - \vone_S^T   \mD(\vx_T-\pi(\vone_S)) \leq  w(\partial S),
\]
As this holds $\forall t\leq T$, we then have:
\begin{align*}
    \vone_S^T   \mD(\vx_T-\pi(\vone_S)) &\geq \vone_S^T   \mD(\vone_S-\pi(\vone_S)) - Tw(\partial S)\\
    &= \frac{\Vol(S)\Vol(\overline{S})}{\Vol(G)} - Tw(\partial S)\\
    &=\Vol(S)\left(\frac{\Vol(\overline{S})}{\Vol(G)} - T\frac{w(\partial S)}{\Vol(S)}\right)\\
    &=\Vol(S)\left(\frac{\Vol(\overline{S})}{\Vol(G)} - T\phi(S)\right)
\end{align*}
In particular,  the assumption $\Vol(S) \leq \Vol(G)/2$ implies $\frac{\Vol(\overline{S})}{\Vol(G)} \geq \frac{1}{2}$, so for $T\phi(S) \leq 1/2$ we have:
\[
    \left(\vone_S^T   \mD(\vx_T-\pi(\vone_S))\right)^2 \geq \Vol(S)^2\left(\frac{1}{2} - T\phi(S)\right)^2
\]  
Chaining the above bounds, we can thus lowerbound our target ratio by
\begin{align*}
    {\norm{\vx_T - \pi(\vx_0)}_{\mD}^2 \over \norm{\vx_0 - \pi(\vx_0)}_{\mD}^2 } &\geq \frac{\frac{1}{\Vol(S)}\left(\vone_S^T   \mD(\vx_0-\pi(\vx_0)) - Tw(\partial S)\right)^2}{\norm{\vx_0 - \pi(\vx_0) }^2_{\mD}}\\
    &\geq \frac{\Vol(G)}{\Vol(S)\Vol(\overline{S})} \Vol(S)\left(\frac{1}{2} -T\phi(S)\right)^2\\
    &\geq \left(\frac{1}{2} -T\phi(S)\right)^2
\end{align*}
where the last line follows from $2 \geq \frac{\Vol(G)}{\Vol(\overline{S})} \geq 1$.
\end{proof}

\newpage

\section{Deferred proofs}\label{appendix:missing-proofs}
In this section, we provide the proofs to all the technical results in the paper.

\subsection{Deferred proofs from Section~\ref{sec:preliminaries}}\label{sec:preliminaries-proof}

\begin{proof}[{\bf Proof of Lemma~\ref{lem:lovasz-is-semi-norm}}]
Let $\delta_h: 2^h \to\R_{\geq 0}$ be a symmetric submodular cut function, i.e., $\delta_h(S) = \delta_h(h\setminus S)$ for all $S \subseteq h$ and $\delta_h(h) = \delta_h(\emptyset) = 0$ and let $\overline{\delta}_h$ be its Lov\'asz extension.

 By standard properties of the Lov\'asz extension (see, e.g., Proposition 3.1 (d)  in  \cite{bachLearningSubmodularFunctions2013}), 
for any $\vx \in \R^h$ and for any $u\in \R$, we have:
\begin{equation}\label{eq:invariance-under-one-shift}
    \bar{\delta}_h(\vx - u \vone_h) = \bar{\delta}_h(\vx) + |u|\cdot \delta_h(h) = \bar{\delta}_h(\vx) .
\end{equation}

Consider now the function $g : \R^n \to \R$ given by:
\[
    g(\vx) \defeq \bar{\delta}_h(\vx) + |\vone_h^T \vx|.
\]
By applying \eqref{eq:invariance-under-one-shift}, we see that for any $\vx\in \R^n$:
\[
    \min_{u\in \R} g(\vx -u\vone_h) = \min_{u\in \R} \bar{\delta}_h(\vx - u \vone_h) + |\vone_h^T (\vx-u\vone_h)| = \min_{u\in \R} \bar{\delta}_h(\vx ) + |\vone_h^T (\vx-u\vone_h)|= \bar{\delta}_h(\vx).
\]
Hence, to prove the lemma it is sufficient to show that $g$ is a norm.
We begin by showing that $g$ is absolutely homogenous, i.e. for every $\vx \in \R^n$ and $\alpha \in \R, \: g(\alpha \cdot \vx) = \alpha \cdot g(\vx)$ . If $\alpha \geq 0$, we can use positive homogeneity of $\bar{\delta}_h$ to obtain:
\begin{align*}
    g(\alpha \cdot\vx) &= \bar{\delta}_h(\alpha\cdot\vx) + |\vone_h^T \alpha\cdot\vx| =  \alpha\cdot\bar{\delta}_h(\vx) + |\alpha|\cdot |\vone_h^T \vx| = \alpha \cdot g(\vx) = |\alpha|\cdot g(\vx).
\end{align*}

Since $\delta_h$ is symmetric, $\bar{\delta}_h$ is even (See e.g. Proposition 3.1 (g) in \cite{bachLearningSubmodularFunctions2013}.) giving for $\alpha <0$:
\begin{align*}
    g(\alpha \cdot\vx) &= \bar{\delta}_h(\alpha\cdot\vx) + |\vone_h^T \alpha\cdot\vx|= \bar{\delta}_h(-\alpha\cdot\vx) + |\vone_h^T \alpha\cdot\vx|\\
    &= -\alpha \cdot\bar{\delta}_h(\vx) +|\alpha| \cdot|\vone_h^T \vx| = |\alpha| \cdot g(\vx)
\end{align*}

giving that $g$ is absolutely homogenous. 

Next, we show that the triangle inequality holds for $g$. To this end we first note that that $g$ is a convex function, since it is the sum of $\bar{\delta}_h$, which is the Lov\'asz extension of a submodular function, and the function $\vx \mapsto |\vone_h^\top \vx|$ which is easily shown to be convex.

For any $\vx, \vy \in \R^h$, we then have:
$$
g(\vx + \vy) = 2 \cdot g\left(\frac{\vx + \vy}{2}\right) \leq g(\vx) + g(\vy),
$$
where the first equality follows from the positive homogeneity of $g$ and the second from its convexity.

Finally, we show that $g$ is positive definite. 
As done in~\cite{li2018submodular}, we assume, without loss of generality\footnote{If $\delta_h(\{v\}) = 0,$ by submodularity and symmetry, we have that $\delta_h(S \cup \{v\}) = \delta_h(S)$ for all $S \subseteq h$. In this case, $v$ can be safely excluded from $h$ without changing the value of the cut function.} that $\delta_h(\{v\}) > 0$ for every $v \in h.$ 
We consider two cases. If $\vx$ is a non-zero multiple of $\vone_h$, we have:
$$
g(\vx) = g(\alpha \cdot \vone_h) = \bar{\delta}_h(\alpha \cdot \vone_h) + |\alpha| \cdot |h| =  |\alpha| \cdot |h| > 0
$$
(Note that here $|\alpha|$ indicates the absolute value of a real number, while $|h|$ indicates the cardinality of a set). If $\vx$ is not a multiple of $\vone_h$, then there exists some $u^* \in \R$ such that $\vx + u^* \vone_h \neq 0$ and $\vone_h^T(\vx + u^* \vone_h)=0$.

We then have, by the definition of Lov\'asz extension and the non-negativity of $\delta_h$:
$$
g(\vx) =\bar{\delta}_h(\vx) + |\vone_h^T \vx| \geq  \bar{\delta}_h(\vx) = \bar{\delta}_h(\vx + u^* \vone_h) \geq \delta_h\left(\argmax_{i \in h}\{ \vx_i + u^*\}\right) > 0.
$$
The last inequality follows from our assumption.

\end{proof}

\begin{proof}[\textbf{Proof of Lemma~\ref{lemma:eigenvalue-type-bounds}}]
The lower bound in Equation~\eqref{eq:eigenvalue-type-bounds}follows by definition of $\lambda_G$. The upper bound relies on Assumption~\ref{assumption}:
    \begin{align*}
        2\cU(\vx) &= \sum_{h\in E} w_h \norm{\vx}_h^2 \leq \sum_{h\in E} w_h \norm{\vx}_2^2 \\
        &= \sum_{i \in V} \sum_{h:h\ni i}w_h \vx(i)^2  = \sum_{i \in V} \deg(i) \vx(i)^2 =\norm{\vx}^2_{\mD}.
    \end{align*}
For the second part of the lemma, assume $\lambda_G = 0.$ Then, there is a non-zero vector $\vx \in \R^V \perp_{\mD} \vone$ with $\cU(\vx) = 0.$ Consider the cut $S \subset V$ consisting of $i\in V$ having $\vx_i > 0.$ Because $\ip{\vx,\vone}_{\mD} = 0,$ the cut $S$ is non-empty. We claim that $S$ constitutes a connected component of $G$ and completes the proof of the lemma. By way of contradiction, suppose a hyperedge $h$ is cut by $S,$ i.e. $h \cap S , h \cap \bar{S} \neq \emptyset$. Then, $\vx$ is not constant over $h$, so that $\vx_h - u\vone_h \neq 0$ for all $u$. Because $\|\cdot \|_h$ is a norm, this implies that hyperedge $h$ makes a positive contribution to $\cU(\vx).$
\end{proof}

\subsection{Deferred proofs from Section~\ref{sec:heat_diffusion}}

To prove Theorem~\ref{thm:characterizing_subgradient_flow}, we leverage classic results from the study of gradient flows. Let $H$ be a Hilbert space with inner product $\langle\cdot,\cdot\rangle$ and norm $\norm{\cdot}$. For a map $A$ let $D(A) \subseteq H$ denote the domain of $A$, i.e. the set $\vx$ such that $A[\vx] \neq \emptyset$. For $\vx \in D(A)$, we do not assume that $A[\vx]$ is single-valued: in general, $A[\vx]\subseteq H$. 
\begin{theorem}[Existence and uniqueness of solutions to maximal-monotone evolutions. Theorems 3.1 and 3.2 in \citet{brezisOperateursMaximauxMonotones1973}. Translation from French by authors.]\label{thm:brezis-key-result}
    Let $H$ be a Hilbert space and $F$ a convex, proper, and lower semicontinuous function on $H$. Then for all $\vx_0 \in D(A)$, there exists a unique function $\vx(t) \in C([0,\infty), H)$ such that
    \begin{enumerate}
        \item $\frac{d}{dt}\vx(t) + \partial F[\vx(t)] \ni 0$ almost everywhere w.r.t. $t$ on $(0, \infty)$
        \item $\vx(0) = \vx_0$
        \item $\partial F[\vx] \neq \emptyset\  \forall t > 0$
        \item $\vx(t)$ is Lipschitz, i.e. $\dot{\vx}(t) \in L^{\infty}(0, \infty; H)$
    \end{enumerate}
    Moreover, such a solution $\vx(t)$ satisfies the following: 
    \[
        \frac{d^+}{dt} \vx(t) + \partial F^0[\vx(t)] = 0  \ \forall t\in [0,\infty)
    \]
    where $\frac{d^+}{dt}$ denotes the right-derivative of $\vx$ at $t$, and $\partial F^0$ denotes the so-called \textit{principal section} of $\partial F$:
    \[
        \partial F^0[\vx] \defeq \min_{\vy \in \partial F[\vx]}\norm{\vy}
    \]
\end{theorem}
We now use this results to prove our theorem:
\begin{proof}[{\bf Proof of Theorem~\ref{thm:characterizing_subgradient_flow}}]
    For $\mD$ positive-definite, $\R^n$ equipped with $\langle\cdot,\cdot\rangle_{\mD}$ and $\norm{\cdot}_{\mD}$ is a Hilbert space. Denote this space $H$. Consider $\cU(\vx)$ as in Equation~\eqref{eq:our-setting} and observe that $\cU(\vx)$ is convex and lower-semicontinuous on $H$. Let $\partial_{\mD} \cU(\vx)$ denote the subdifferential of $\cU$ with respect to the geometry of $H$.
    
    Then, by Theorem~\ref{thm:brezis-key-result}, for any $\vx_0\in \R^n$ there exists a unique solution $\vx(t) \in C([0,\infty);H)$ with $\dot{\vx}(t)\in L^{\infty}(0,\infty; H)$ such that
    \[
        \begin{cases}
            \vx(0)=  \vx_0\\
            \dot{\vx}(t) \in -\partial_{\mD}\cU(\vx(t)) \text{ almost everywhere for } t\geq 0
        \end{cases}
    \]
    and that for all $t\in [0,\infty)$,
    \[
        \lim_{h\rightarrow 0}\frac{\vx(t+h)-\vx(t)}{h}= -\argmin_{\vy\in \partial_{\mD}\cU(\vx(t))}\norm{\vy}_{\mD}.
    \]
    We now examine $\partial_{\mD} \cU(\vx)$. Observe that by defnition,
    \begin{align*}
        \partial_{\mD} \cU(\vx) &= \{\vy\in \R^n \mid \cU(\vz) \geq \cU(\vx) + \ip{\vz-\vx,\vy}_{\mD}, \forall \vz \in \R^n\}\\
        &= \{\mD^{-1}\vy\in \R^n \mid \cU(\vz) \geq \cU(\vx) + \ip{\vz-\vx,\vy}, \forall \vz \in \R^n\}\\
        &= \{\mD^{-1}\vy\in \R^n \mid \vy\in \cL(\vx)\} = \mD^{-1}\cL(\vx)
    \end{align*}
    as we define $\cL(\vx)$ to be the subgradient of $\cU$ with respect to the Euclidean inner product. Thus, the solution $\vx(t)$ satisfies
    \[
        \begin{cases}
            \vx(0)=  \vx_0\\
            \dot{\vx}(t) \in -\mD^{-1}\cL(\vx(t)) \text{ almost everywhere for } t\geq 0
        \end{cases}
    \]
    and that a.e. on $t\in [0,\infty)$,
    \[
        -\dot{\vx}(t) = \argmin_{\vy\in \mD^{-1}\cL(\vx)}\norm{\vy}_{\mD} = \argmin_{\vy \in \cL(\vx(t))} \norm{\vy}_{\mD^{-1}}
    \]
    Finally, we observe that for $\mD$ positive definite and $H$ as defined, the function classes $C([0,\infty);H)$ and $L^{\infty}(0,\infty; H)$ are equivalent to the classes $C([0,\infty);\R^n)$ and $L^{\infty}(0,\infty; \R^n)$ on Euclidean space with respect to the Euclidean norm and inner product. Thus the result holds.
\end{proof}

\begin{proof}[\textbf{Proof of Theorem~\ref{thm:cts-time-diffusion-convergence}}]
For $\vx(t)$ continuous on $[0,\infty)$, let $\dot{\vx}(t)\in \R^n$ denote the vector whose entries are the derivatives of $\vx$ with respect to time. For $\vx(t)$ satisfying $\dot{\vx}(t) \in -\mD^{-1}\cL(\vx(t))$  a.e. on $[0,T]$, by Property~\ref{property:pi-unchanging-cts-time} we have
\begin{align*}
    \frac{d}{dt}\left(\frac{1}{2}\norm{\vx(t) - \pi(\vx_0)}^2_{\mD}\right) &=
    \frac{d}{dt}\left(\frac{1}{2}\norm{\vx(t)}^2_{\mD}\right)\\
    &= -\langle \mD^{-1}\cL(\vx(t)), \vx(t)\rangle_{\mD}\\
    &=-2\cU(\vx(t))
\end{align*}
thus establishing the instantaneous result.

For the aggregate result, we observe that by Property~\ref{property:pi-unchanging-cts-time} and the definition of $\pi$,
\[
    \vx(t) - \pi(\vx(t)) = \vx(t) - \pi(\vx_0) \quad \text{ and }\quad \left(\vx(t) - \pi(\vx_0)\right)\perp_{\mD}\vone.
\]
Moreover, by definition $\pi(\vx_0)$ parallel to $\vone$. Thus, since $\cU$ is invariant under additive shifts of $\vone$,
\begin{align*}
    -2\cU(\vx(t)) &= -2\cU(\vx(t)-\pi(\vx_0))\\
    &\leq -\lambda_G\norm{\vx(t)-\pi(\vx_0)}^2_{\mD}\\
    &=-2\lambda_G \left(\frac{1}{2}\norm{\vx(t)-\pi(\vx_0)}^2_{\mD}\right),
\end{align*}
by Corollary~\ref{cor:bounded-subgradient-norm} and the definition of the isoperimetric constant. $\vx(t) \in C([0,\infty),\R^n)$ implies continuity of $\norm{\vx(t)-\pi(\vx_0)}_{\mD}$, so by Grönwall's inequality the above bound implies
\[
    \frac{1}{2}\norm{\vx(t)-\pi(\vx_0)}^2_{\mD} \leq \frac{1}{2}\norm{\vx_0-\pi(\vx_0)}^2_{\mD}\exp\left(\int_0^t -2\lambda_G ds\right) = \frac{1}{2}\norm{\vx_0-\pi(\vx_0)}^2 e^{-2\lambda_G t}.
\]
\end{proof}
\begin{proof}[\textbf{Proof of Theorem~\ref{thm:discrete-time-diffusion}}]
    By the definition of the iterates:
    \begin{align*}
        \norm{\vx_{t+1}-\pi(\vx_0)}_\mD^2 & = \norm{\vx_t -  \eta\mD^{-1} \cLD (\vx_t) - \pi(\vx_0)}_\mD^2\\
        &= \norm{\vx_t-\pi(\vx_0)}^2_\mD - 2\eta \ip{\mD^{-1}\cLD (\vx_t), \vx_t - \pi(\vx_0)}_{\mD} + \eta^2 \norm{\mD^{-1}\cLD(\vx_t)}_\mD^2 \\
        &= \norm{\vx_t-\pi(\vx_0)}^2_\mD - 4\eta \cU(\vx_t) + \eta^2 \norm{\cLD(\vx_t)}_{\mD^{-1}}^2
    \end{align*}
    where the last line follows from Property~\ref{prop:inner-prod-with-cL} and the fact that by Equation~\ref{eq:laplacian} $\cL(\vx) \perp\vone$ for all $\vx$. Applying Corollary~\ref{cor:bounded-subgradient-norm} to the last term, we obtain
    \[
        \norm{\vx_{t+1}-\pi(\vx_0)}_\mD^2 \leq \norm{\vx_t - \pi(\vx_0)}^2_{\mD} + \cU(\vx_t)\left(2\eta^2 - 4\eta\right),
    \]
    This bound is extremized with respect to $\eta$ when $\eta=1$, yielding:
    \[
        \norm{\vx_{t+1} - \pi(\vx_0)}_\mD^2 \leq \norm{\vx_t-\pi(\vx_0)}^2_{\mD} - 2\cU(\vx_t) 
    \]
    which thus establishes the ``instantaneous" property.
    
    For the aggregate property, observe that $\pi(\vx_0)$ is parallel to $\vone$ and that $(\vx_t - \pi(\vx_0)) \perp_{\mD} \vone$. Thus by Lemma~\ref{lemma:eigenvalue-type-bounds}
    \[
       -2\cU(\vx_t) = -2\cU(\vx_t-\pi(\vx_0)) \leq -\lambda_G \norm{\vx_t-\pi(\vx_0)}.
    \]
    We thus have
    \[
        \norm{\vx_{t+1}-\pi(\vx_0)}_\mD^2 \leq (1-\lambda_G) \norm{\vx_t-\pi(\vx_0)}^2,
    \]
    and since this holds for all $t\in\mathbb{N}_{\geq 0}$,
    \[
        \norm{\vx_{t+1}-\pi(\vx_0)}_\mD^2 \leq (1-\lambda_G)^t \norm{\vx_0 - \pi(\vx_0)}^2_{\mD}.
    \]
\end{proof}

\subsection{Deferred proofs from Section~\ref{sec:resolvent}}

\begin{proof}[{\bf Proof of Lemma~\ref{lem:hypergraph-pagerank-fixed-point}}]
Letting $\vx = \mD^{-1} \vp_{\alpha}(\vs)$ we see that:
\[
    \vs- \vp_{\alpha}(\vs) \in {1-\alpha \over 2\alpha} \cL (\mD^{-1} \vp_{\alpha}(\vs)) 
\]
\[
    \iff \vs - \mD \vx \in {1-\alpha \over 2\alpha}\cL (\vx) 
\]
\[
    \iff { 2\alpha \over 1-\alpha}\vs - { 2\alpha \over 1-\alpha}\mD \vx \in \cL (\vx) 
\]
giving $\vx \in \cR_{\nicefrac{2\alpha}{1-\alpha}}((\nicefrac{2\alpha}{1-\alpha})\vs)$ as needed.
\end{proof}

\begin{proof}[{\bf Proof of Theorem~\ref{thm:main1}}]
The only obstacle to applying Theorem~\ref{thm:optimization} directly is the fact that the optimization in Problem~\ref{eq:resolvent} is over $\R^V,$ while Lemma~\ref{lemma:energy-is-seminorm} only guarantees that the objective is a norm over $\R^V \perp_{\mD} \vone.$
To bypass this obstacle, we show how the solution to the resolvent optimization problem over $\R^V$ can be reduced to that over $\R^V \perp_{\mD} \vone.$ To this end, we decompose a solution $\vx$ into its component along $\vone$ and its component orthogonal to $\vone,$ with respect to the $\mD$-inner product.
$$
\vx = \vx_{\perp} + \pi(\vx) = \vx_{\perp} + \frac{\ip{\vx,\vone}_{\mD}}{\|\vone\|_{\mD}^2} \vone.
$$
By its definition, $\cU(\vx)$ is invariant under shifts by multiples of $\vone$, so that we can decompose the optimization problem~\ref{eq:resolvent} into two separate optimization problems:
\begin{align} 
\nonumber \textrm{OPT} \defeq  \min
_{\vx \in \R^V} 
\cU(\vx) + \frac{\lambda}{2} \|\vx\|^2_{\mD} - \ip{\vs,\vx} =\\
\nonumber \min_{\substack{\vx_\perp \in \R^V \perp_{\mD} \vone\\ \alpha \in \R}} \cU(\vx_\perp + \alpha \vone ) + \frac{\lambda}{2} \cdot \|\vx_{\perp} + \alpha \vone\|^2_{\mD} - \ip{\vs,\vx_\perp + \alpha \vone}=\\
\label{eq:opt-equiv}
\min_{\vx_\perp \in \R^V \perp_{\mD} \vone} \cU(\vx_\perp) + \frac{\lambda}{2} \cdot\|\vx_{\perp}\|^2_{\mD} - \ip{\vs,\vx_\perp}+ \min_{\alpha \in \R} \frac{\lambda}{2} \cdot \alpha^2 \cdot  \|\vone\|^2_{\mD}- \alpha \cdot \ip{\vs, \vone} 
\end{align}

If $\lambda = 0$ and $\ip{\vs, \vone} \neq 0,$ then the optimization problem is unbounded, which can be easily detected by our algorithm.
If $\lambda = 0$ and $\ip{\vs, \vone} = 0,$ we have successfully reduced the problem to the form of Theorem~\ref{thm:optimization}. Hence, we may assume that $\lambda > 0.$

The second minimization is optimized by choosing $\alpha = \nicefrac{ \ip{\vs, \vone}}{\lambda \cdot \|\vone\|^2_{\mD}},$ which yields:
\begin{align*}
\textrm{OPT} =
\min_{\vx_\perp \in \R^V \perp_{\mD} \vone} \cU(\vx_\perp) + \frac{\lambda}{2} \cdot\|\vx_{\perp}\|^2_{\mD} - \ip{\vs,\vx_\perp} - \frac{\ip{\vs, \vone}^2}{\lambda\cdot \|\vone\|^2_{\mD}} = \textrm{OPT}_\perp -  \frac{\ip{\vs, \vone}^2}{\lambda\cdot \|\vone\|^2_{\mD}},
\end{align*}
where $\textrm{OPT}_\perp$ denotes the optimum of the minimization problem restricted to $\R^V \perp_{\mD} \vone.$
By our assumptions on $\cU$, we have the following norm bounds:
$$
\forall\, \R^V \perp_{\mD} \vone, 
\; \frac{(\lambda_G + \lambda_2)}{2} \cdot \|\vx\|^2_{\mD}   \leq \cU(\vx) + \frac{\lambda}{2} \|\vx\|^2_{\mD} \leq \frac{(1 + \lambda)}{2} \cdot  \|\vx\|^2_{\mD}.
$$
Hence, by Theorem~\ref{thm:optimization}, after $O\left(\nicefrac{1}{(\lambda_G + \lambda) \epsilon^2}\right)$ iterations, Algorithm~\ref{alg:opt} outputs $\vx_{\perp}$ such that
$$
\cU(\vx_\perp) + \frac{\lambda}{2} \cdot\|\vx_{\perp}\|^2_{\mD} - \ip{\vs,\vx_\perp} \leq \textrm{OPT}_{\perp} + \epsilon \cdot |\textrm{OPT}_{\perp}| 
$$
Finally, the algorithm for Problem~\ref{eq:resolvent} will return:
$$
\vx_{\textrm{out}} = \vx_\perp + \frac{ \ip{\vs, \vone}}{\lambda \cdot \|\vone\|^2_{\mD}} \vone.
$$
By Equation~\ref{eq:opt-equiv}, we have:
\begin{align*}
\cU(\vx_{\textrm{out}}) + \frac{\lambda}{2} \cdot \|\vx_{\textrm{out}}\|^2_{\mD} - \ip{\vs,\vx_{\textrm{out}}} = \cU(\vx_\perp) + \frac{\lambda}{2} \cdot\|\vx_{\perp}\|^2_{\mD} - \ip{\vs,\vx_\perp}  -  \frac{\ip{\vs, \vone}^2}{\lambda\cdot \|\vone\|^2_{\mD}}\leq \\
\textrm{OPT}_{\perp} + \epsilon \cdot |\textrm{OPT}_{\perp}|  -  \frac{\ip{\vs, \vone}^2}{\lambda\cdot \|\vone\|^2_{\mD}} = \textrm{OPT} + \epsilon \cdot |\textrm{OPT}_{\perp}| \leq
\textrm{OPT} + \epsilon \cdot |\textrm{OPT}| 
\end{align*}
This proves the approximation result and the bound on the number of iterations. Each iteration requires updating a vector by $\eta \mD^{-1} \vs$, which requires linear time in the size of the hypergraph, and updating in the heat diffusion direction $\eta \mD^{-1} \cL(\vx_t)^\mD.$
\end{proof}

\subsection{Deferred proofs from Section~\ref{sec:optimization}}\label{appendix:proof-of-optimization-convergence}

\begin{proof}[\textbf{Proof of Lemma~\ref{lemma:energy-is-seminorm}}]
By Lemma~\ref{lem:lovasz-is-semi-norm}, each $\|\cdot\|_h$ is a norm over $\R^h$ and hence a semi-norm over $\R^V.$ 
We proceed to check the absolute homogeneity of $\cU(\cdot) + \nicefrac{\lambda}{2} \|\cdot \|^2_{\mD}$. For all $\vx \in \R^V$ and $a \in \R$:
\begin{align*}
\sqrt{2 \cdot\cU(a\vx) + \lambda \norm{a\vx}^2_{\mD}} &= \sqrt{\sum_{h\in E}w_h \min_{u\in \R}{\norm{a\vxh - u\1}_h^2}+ \lambda a^2\norm{\vx}^2_{\mD}}\\
&= \sqrt{a^2\sum_{h\in E}w_h \min_{u\in \R}{\norm{\vxh - u\1}_h^2} + \lambda a^2\norm{\vx}^2_{\mD}}\\
&= |a| \sqrt{\sum_{h\in E}w_h \min_{u\in \R}{\norm{\vxh - u\1}_h^2}+ \lambda \norm{\vx}^2_{\mD}}\\
&= |a| (\sqrt{2 \cdot\cU(\vx) + \lambda \norm{\vx}_{\mD}^2}),
\end{align*}
Next, we establish subadditivity. For all $\vx, \vy \in \R^V$:
\begin{align*}
\sqrt{2 \cdot\cU(\vx + \vy) + \lambda \norm{\vx+\vy}_{\mD}^2} &= \sqrt{\sum_{h\in E}w_h \min_{u\in \R}{\norm{\vxh + \vyh - u\1}_h^2} + \lambda\norm{\vx+\vy}_{\mD}^2}\\
&\leq \sqrt{\sum_{h\in E}w_h \min_{u\in \R}{\norm{\vxh - u\1}_h^2} + \lambda\norm{\vx}_{\mD}^2} + \sqrt{\sum_{h\in E}w_h \min_{u\in \R}{\norm{\vyh - u\1}_h^2}+\lambda\norm{\vy}_{\mD}^2}\\
&= \sqrt{2 \cdot\cU(\vx) +\lambda\norm{\vx}_{\mD}^2} + \sqrt{2\cdot\cU(\vy)+\lambda\norm{\vy}_{\mD}^2}
\end{align*}
In the above, the last part follows from subadditivity of $\min_{u\in\R} \norm{\cdot - u\1}_h$ function, together with the observation that $\sqrt{U(\cdot)}$ is the $\norm{\cdot}_{\diag(\vw)}$ norm of the vector of $h$-seminorms, and:
\[
    \norm{\vv}_{\diag(\vw)} \leq \norm{\vw}_{\diag(\vw)},
\]
whenever $\vv \leq \vw$ entry-wise. 

Finally, the proposed function is positive definite, because the connectedness of $\lambda_G$ implies, by Lemma~\ref{lemma:eigenvalue-type-bounds}, that
$$
\forall \vx \perp_{\mD} \vone\: , \cU(\vx) = 0 \iff \|\vx||^2_{\mD} = 0 \iff \vx = 0.
$$
\end{proof}

\subsubsection{Proof of Theorem~\ref{thm:optimization}}

We will need the following lemma, which captures an important consequence of any upper bound of the form of Equation~\ref{eq:eigenvalue-type-bounds}:
\begin{corollary}\label{cor:bounded-subgradient-norm}
Let $F(\cdot) = \nicefrac{1}{2} \|\cdot\|^2$ for a semi-norm $\|\cdot\|$ and $\mR$ be a positive semi-definite operator. Suppose there exists $u \in \R_{>0}$ such that, for all $\vx$, $F(\vx) \leq \nicefrac{u}{2} \|\vx\|^2_{\mR}.$ Then:
$$
\vz \in \partial F(\vx) \; \mathrm{ implies } \; \frac{1}{2}\|\vz\|^2_{\mR^{-1}} \leq u \cdot F(\vx)
$$
In particular, Equation~\ref{eq:eigenvalue-type-bounds} implies that for all $\vz\in \cL(\vx)$, $\nicefrac{1}{2} \cdot \norm{\vz}^2_{\mD^{-1}} \leq \cU(\vx)$.
\end{corollary}

\begin{proof}[\textbf{Proof of Corollary~\ref{cor:bounded-subgradient-norm}}]
Taking the Fenchel dual of $\norm{\cdot}_{\mR^{-1}}$,
\begin{align*}
    {1\over 2}\norm{\vz}_{\mR^{-1}}^2&=\max_{\vy} \ip{\vz, \vy} -{\norm{\vy}_{\mD}^2 \over 2}\\
    &\leq\max_{\vy} \ip{\vz, \vy} - \frac{F(\vy)}{u}
\end{align*}
This  expression is maximized with respect to $\vy$ when $u \cdot \vz \in \partial F(\vy)$. Because $\partial F$ is linear homogeneous, this is satisfied for $\vy = u \cdot \vx$, so
\begin{eqnarray*}
    {1\over 2}\norm{\vz}_{\mR^{-1}}^2  &\leq& u \cdot \ip{\vz, \vx} - \frac{F(u \cdot \vx)}{u}\\ &=& u \cdot (\ip{\vz, \vx} - F(\vx)) = u \cdot F(\vx)
\end{eqnarray*}
The last step follows by Property~\ref{fct:doublenorm}.
\end{proof}

The first step of the proof of the Theorem~\ref{thm:optimization} closely follows the standard analysis of optimistic mirror descent, where the optimistic part of the step is performed with respect to the fixed part of the subgradient: $\vs$. We summarize it in the following lemma:
\begin{lemma}\label{lem:optimistic-md} For all $t \in \N,$
\begin{align*}
(& F(\hat{\vx}_t) -  \ip{s,\hat{\vx}_t}) - \mathrm{OPT}_{F,\vs} \\ &\leq  \frac{1}{2\eta}\norm{\vx_t - \vx^*}^2_{\mR} - \frac{1}{2\eta}\norm{\vx_{t+1} - \vx^*}^2_{\mR}  + \frac{\eta}{2} \norm{\vz_t}^2_{\mR^{-1}}.\\
\end{align*}
\end{lemma}
\begin{proof}[\textbf{Proof of Lemma~\ref{lem:optimistic-md}}]
For any $t\in \N$, we have:
\begin{align*}
    \frac{1}{2}\norm{\vx_{t+1}-\vx^*}^2_{\mR} &=\frac{1}{2}\norm{\hat{\vx}_t - \eta \mR^{-1}\vz_t- \vx^*}^2_{\mR}\\
    &= \frac{1}{2}\norm{\hat{\vx}_t - \vx^*}_{\mR}^2 -\eta\ip{\hat{\vx}_t - \vx^*, \vz_t} +{\eta^2 \over 2} \norm{  \mR^{-1}\vz_t }^2_{\mR}\\
    &= \frac{1}{2}\norm{\hat{\vx}_t - \vx^*}_{\mR}^2 -\eta\ip{\hat{\vx}_t - \vx^*, \vz_t} +{\eta^2 \over 2} \norm{\vz_t }^2_{\mR^{-1}}\\
    &= \frac{1}{2}\norm{\vx_t + \eta\mR^{-1} \vs -  \vx^*}_{\mR}^2 -\eta\ip{\hat{\vx}_t - \vx^*, \vz_t} +{\eta^2 \over 2} \norm{ \vz_t }^2_{\mR^{-1}}\\
    &=\frac{1}{2}\norm{\vx_t - \vx^*}^2_{\mR} - \eta\ip{\hat{\vx}_t - \vx^*, \vz_t - \vs } + \frac{\eta^2}{2}\left( \norm{\vz_t}^2_{\mR^{-1}}-\norm{\vs}^2_{\mR^{-1}} \right)\\
    &\leq\frac{1}{2}\norm{\vx_t - \vx^*}^2_{\mR} + \eta\left[\textrm{OPT}_{F,\vs} - (F(\hat{\vx}_t) - \ip{s,\hat{\vx}_t})\right] + \frac{\eta^2}{2} \norm{\vz_t}^2_{\mR^{-1}}\\
\end{align*}    
where the last inequality follows from the convexity of $F$ and the positivity of $\|s\|^2_{\mR^{-1}}.$ The lemma follows from re-arranging terms.
\end{proof}

Next, as we do not have an uniform bound on the norm $\norm{\vz_t}_{\mR^{-1}}$ of the subgradient $\vz_t$, we rely on Equation~\ref{eq:poincare}, Corollary~\ref{cor:bounded-subgradient-norm}, and the setting of $\eta$ in Algorithm~\ref{alg:opt} to obtain:
$$
\frac{1}{2} \norm{\vz_t}^2_{\mR^{-1}} \leq u_{\mR} \cdot F(\hat{\vx}_t)\leq \frac{\epsilon}{2\eta} \cdot F(\hat{\vx}_t)
$$
Substituting this equation in the result of Lemma~\ref{lem:optimistic-md} and summing over all $t \in \{0,1,\ldots, T-1\}$ while telescoping terms yields:
\begin{align*}
\sum_{t=0}^{T-1} (1 - \nicefrac{\epsilon}{2}) & \cdot F(\hat{\vx}_t) - \ip{s,\hat{\vx}_t} \\ & \leq T \cdot \mathrm{OPT}_{F,\vs} +  \frac{1}{2\eta}\norm{\vx_{0} - \vx^*}^2_{\mR}
\end{align*}
We now consider the left-hand side and right-hand side of this equation separately. We relate the left-hand side to the objective value at the output point $\vx^{\textrm{out}}_T$. By Jensen's inequality and  the quadratic homogeneity of $F$, we have:
\begin{align*}
\sum_{t=0}^{T-1} (1-\nicefrac{\epsilon}{2}) F(\hat{\vx}_t) - \ip{\vs, \hat{\vx_t}} \geq\\ T \cdot \left((1-\nicefrac{\epsilon}{2}) F(\bar{\vx}_T) - \ip{\vs, \bar{\vx}_t}\right)=\\
\frac{T}{1-\nicefrac{\epsilon}{2}} \cdot \left(F(\vx^{\textrm{out}}_T) - \ip{\vs, \vx^{\textrm{out}}_T}\right) 
\end{align*}

The last term on the righthand side can be easily bounded by the lower inequality in Equation~\ref{eq:poincare} and choice of $\vx_0$ in Algorithm~\ref{alg:opt}:
$$
\frac{1}{2}\norm{\vx_{0} - \vx^*}^2_{\mR} = \frac{1}{2}\norm{ \vx^*}^2_{\mR} \leq \frac{1}{\ell_{\mR} } F(\vx^*) = \frac{|\mathrm{OPT}_{F,\vs}|}{\ell_{\mR}}
$$
Combining our bounds for the left- and right-hand sides, multiplying by $(1-\nicefrac{\epsilon}{2})T^{-1}$, and recalling that $\mathrm{OPT}_{F,\vs} \leq 0$, we derive the  bound:
\begin{align*}
\left(F(\vx^{\textrm{out}}_T) - \ip{\vs, \vx^{\textrm{out}}_T}\right)  \leq (1-\nicefrac{\epsilon}{2})\left( \mathrm{OPT}_{F,\vs} + \frac{|\mathrm{OPT}_{F,\vs}|}{\eta \ell_{\mR} T}\right) \\
\leq \mathrm{OPT}_{F,\vs} + \left(\frac{\epsilon}{2} +\frac{u_{\mR}}{\ell_{\mR}} \cdot \frac{2}{\epsilon T} \right) \cdot |\mathrm{OPT}_{F,\vs}|
\end{align*}
Finally, setting $T \geq \frac{u_{\mR}}{\ell_{\mR}} \cdot \frac{4}{\epsilon^2}$ yields the required multiplicative bound completing the proof of Theorem~\ref{thm:optimization}.

\newpage
\section{Empirical evaluation of hypergraph resolvent algorithm}\label{sec:empirical-comparison}
In this section, we provide a preliminary empirical evaluation of the performance of our algorithm for hypergraph resolvent computation, as stated in Theorem~\ref{thm:main1}. We focus in particular on the task of computing hypergraph PageRank vectors, which has the most practical relevance at this time and has already been studied by~\citet{takai2020hypergraph} and~\citet{liu2021strongly}.    

\subsection{Qualitative Comparison of Graph and Hypergraph PageRank}

We begin by showing a visualization of graph and hypergraph PageRank (with the standard hypergraph potential $\cU_\infty$) to compare their behavior on an easy-to-visualize random geometric graph in Figure~\ref{fig:geometric_KNN}. This is a good setting for comparison as there is a natural way to define both a graph and a hypergraph from the point cloud, by either connecting each vertex $v$ to the $k$-nearest neighbors or including those $k$-neighbors in a hyperedge together with $v.$  

The most salient feature on display here is the fact that the hypergraph PageRank tends to mix faster within clusters than the graph PageRank. This is expected in general as it is possible to prove that, in this case, hypergraph spectral gap is at least as large as the graph spectral gap~\cite{olesker-taylorGeometricBoundsFastest2021}. We believe that this property may be useful in practical clustering tasks, as already evidenced in the manifold learning experiments in Section~\ref{sec:experiments}.

\begin{figure}[ht]
    \begin{subfigure}[b]{0.45\textwidth}
        \centering
        \includegraphics[width=\textwidth]{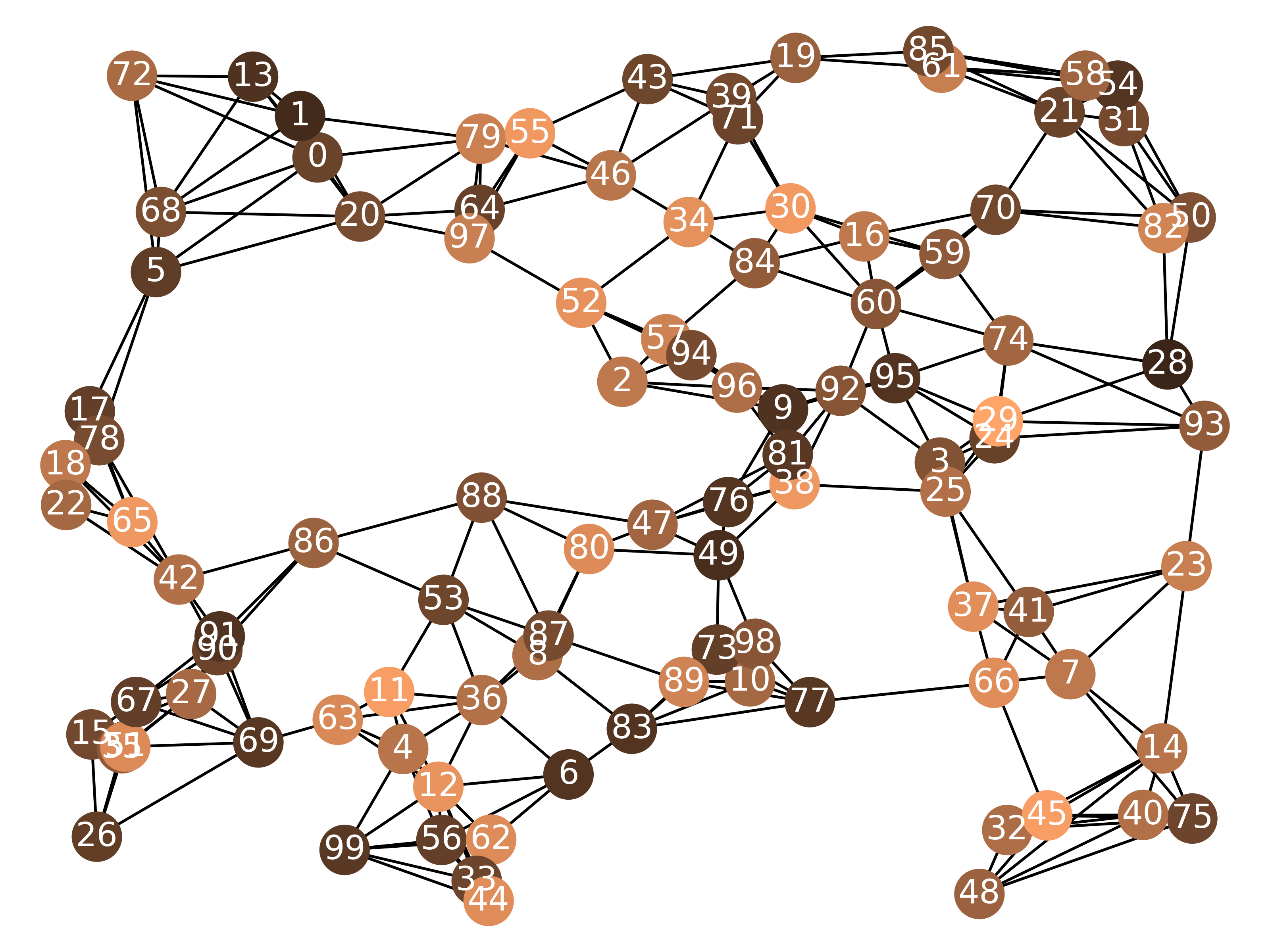}
        \caption{Graph PageRank}
        \label{fig:geometric_graph_PR}
    \end{subfigure}
    \hfill
    \begin{subfigure}[b]{0.45\textwidth}
        \centering
        \includegraphics[width=\textwidth]{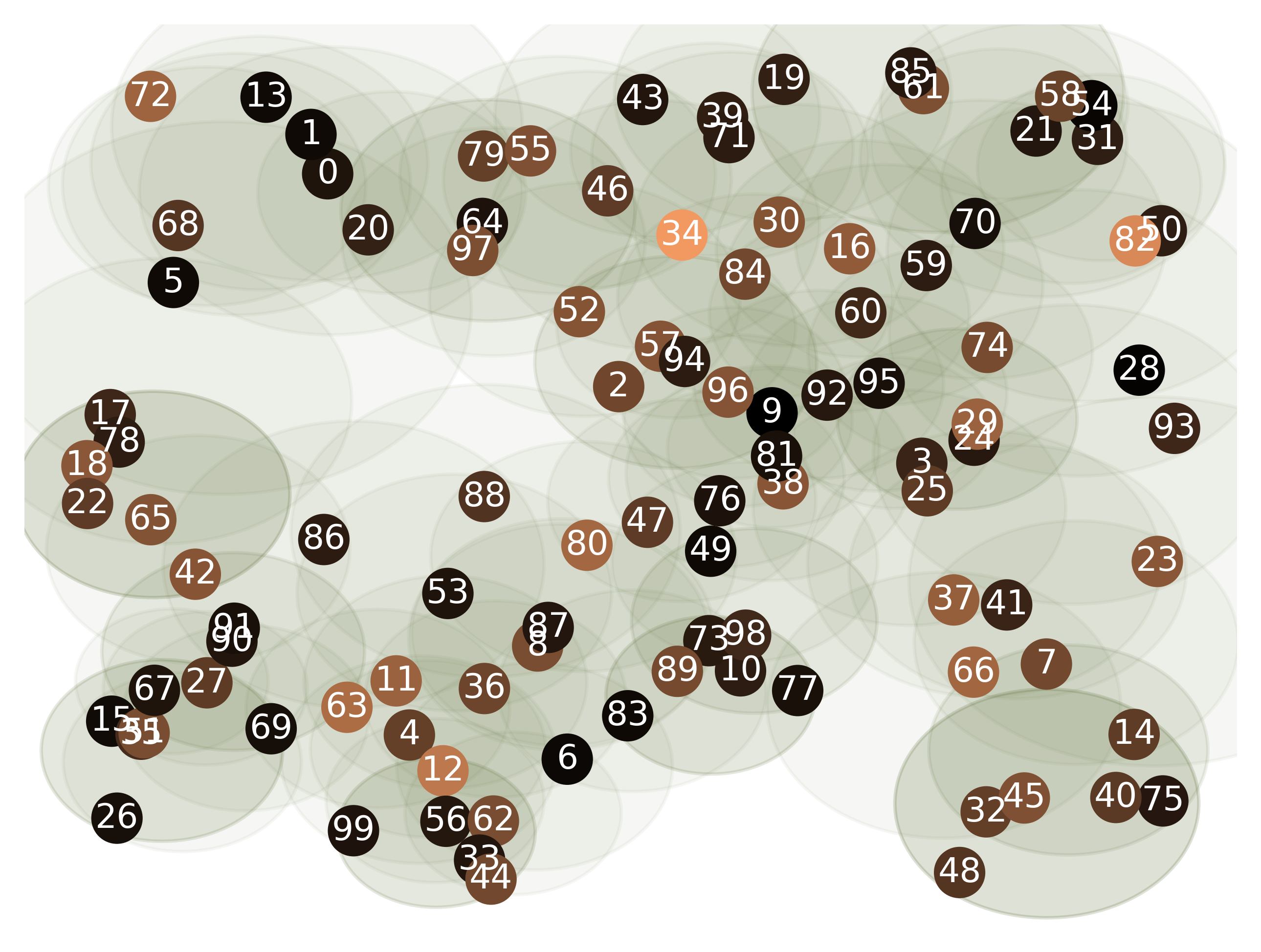}
        \caption{Hypergraph PageRank}
        \label{fig:geometric_hypergraph_PR}
    \end{subfigure}\\
    \begin{subfigure}[b]{0.45\textwidth}
        \centering
        \includegraphics[width=\textwidth]{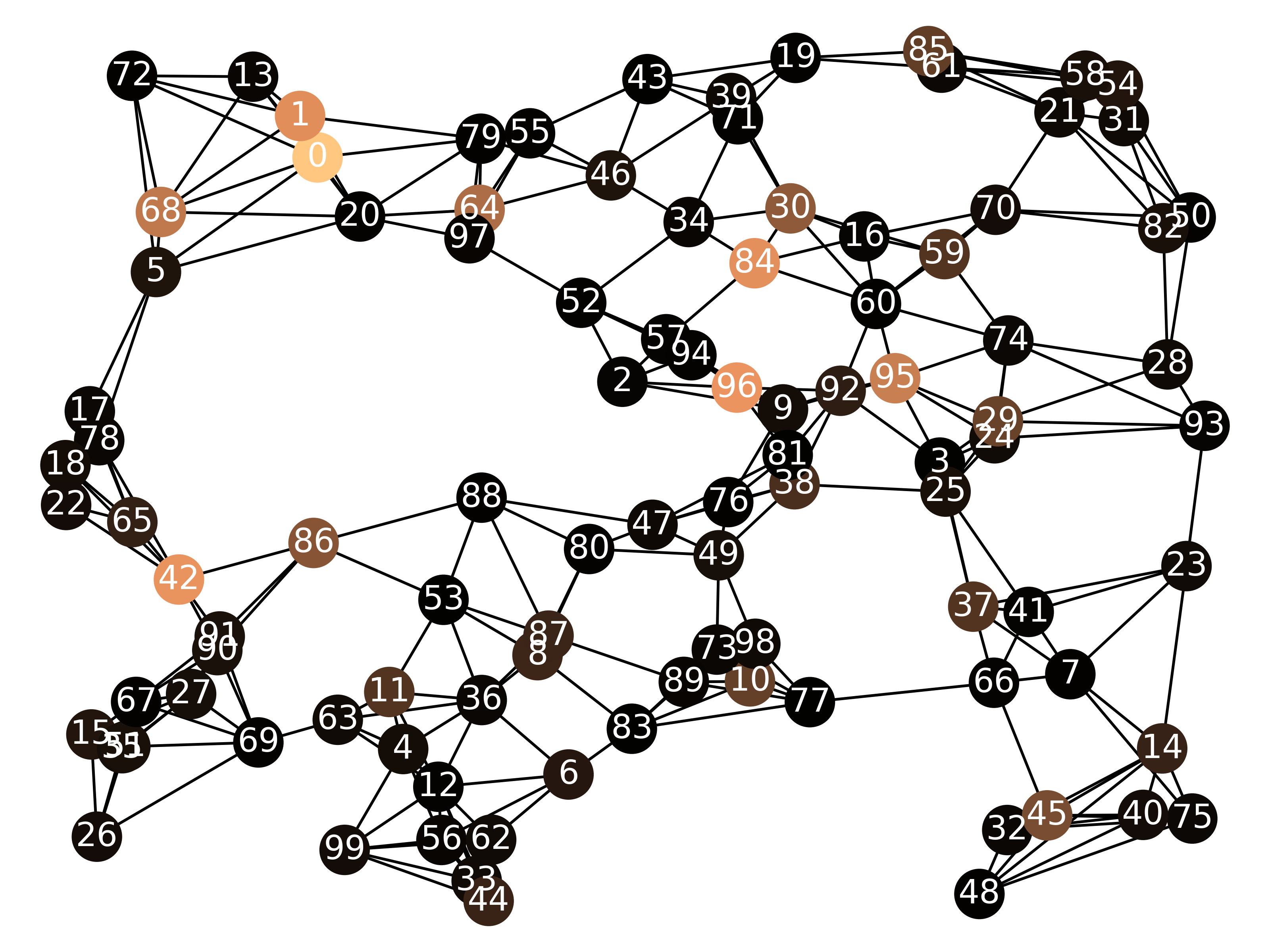}
        \caption{Graph Pesonalized Pagerank}
        \label{fig:geometric_graph_PPR}
    \end{subfigure}
    \hfill
    \begin{subfigure}[b]{0.45\textwidth}
        \centering
        \includegraphics[width=\textwidth]{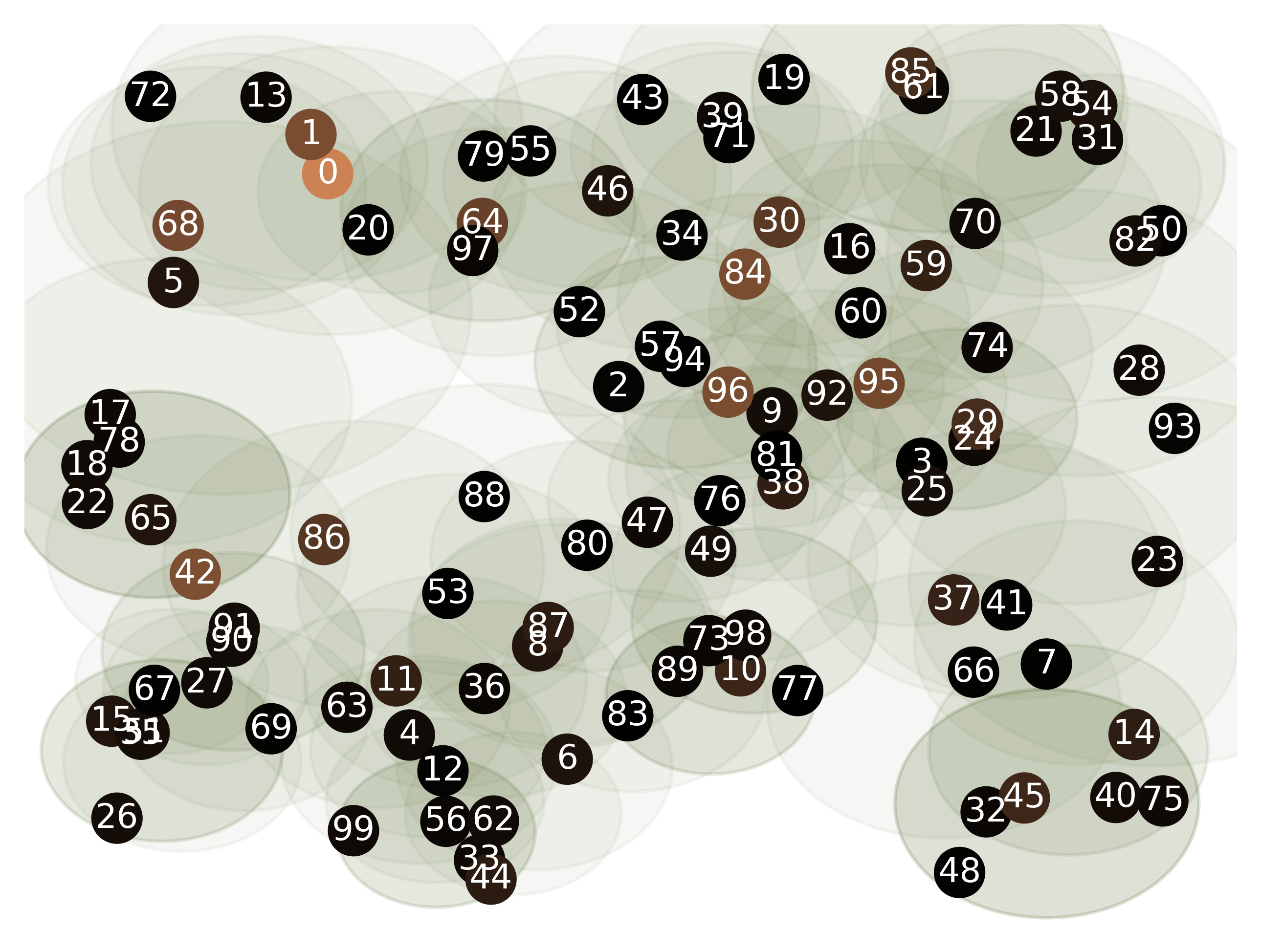}
        \caption{Hypergraph Personalized PageRank}
        \label{fig:geometric_hypergraph_PPR}
    \end{subfigure}
    \caption{ A sample of $n=100$ points uniformly distributed over the unit square. In (\ref{fig:geometric_graph_PR}) and (\ref{fig:geometric_graph_PPR}), we consider the $5$-nearest-neighbor graph on this point cloud, while in (\ref{fig:geometric_hypergraph_PR}) and (\ref{fig:geometric_hypergraph_PPR}) we build a hypergraph by adding a hyperedge for each vertex $v$ containing $v$ and the $v$'s $5$-nearest neighbors.
    Subfigures ((\ref{fig:geometric_graph_PR}) and (\ref{fig:geometric_hypergraph_PR}) show a heat map of the result of computing PageRank with $\lambda = 0.5$ on a fixed random seed.
    Similarly, subfigures (\ref{fig:geometric_graph_PPR}) and (\ref{fig:geometric_hypergraph_PPR}) show the heat map of PageRank $\lambda=0.01$ seeded at vertex $0.$}
    \label{fig:geometric_KNN}
\end{figure}

\subsection{Empirical Evaluation}

We used a variety of real datasets to practically evaluate our methods against existing work. They are presented in 
Table~\ref{tab:read_datasets}.\href{http://konect.cc/networks/}{KONECT} were originally bipartite, unweighted graphs. In order to turn them into hypergraphs, we kept the left-hand-side nodes and replaced each of the right-hand-side nodes with a hyperedge over all its neighbors. If multiple hyperedges are over the same nodes, the corresponding weight is increased. In the Fauci Email dataset, nodes are people and hyperedges are email chains connecting those people. Emails with more than 25 participants have been excluded as those tend to be announcements that do not show an actual connection between people and distort the local structure that we are trying to recover. Only senders and recipients are included and not CCs. DBLP is a co-authorship network in Computer Science  from the largest conferences in the areas of Theory, Data Management - Databases, Data Mining , Vision, Machine Learning and Networks.

\begin{table}[ht]
    \centering
    \begin{tabular}{lcrrrr}
    \hline
        {\bf Name} & {\bf Type} & {\bf |V|} & {\bf |E|} & $\mathbf{\sum\limits_{v \in V} \frac{\deg_v}{|V|}}$ & $\mathbf{\sum\limits_{e \in E} \frac{|e|}{|E|}}$ \\
        \hline
        \href{http://konect.cc/networks/dimacs10-netscience/}{Network Theory}~\cite{kunegis2013konect,newman2006finding} & Co-authorship & 213 & 282 & 3.99 & 2.81 \\
        Fauci Email~\cite{Benson-2021-fauci-emails,Leopold-2021-fauci-emails} & Email Network & 932 & 1283 & 4.00 & 2.91 \\
        \href{http://konect.cc/networks/opsahl-collaboration/}{Arxiv Cond Mat}~\cite{kunegis2013konect,newman2001structure} & Co-authorship & 13861 & 13571 & 3.87 & 3.05 \\
        \href{http://konect.cc/networks/dbpedia-writer/}{DBPedia Writers}~\cite{kunegis2013konect,auer2007dbpedia} & Co-writing credits & 54909 & 18232 & 1.80 & 5.29 \\
        \href{https://projects.csail.mit.edu/dnd/DBLP/}{DBLP} & Co-authorship & 83114 & 74989 & 3.20 & 3.55 \\
        \href{http://konect.cc/networks/youtube-groupmemberships/}{Youtube}~\cite{kunegis2013konect,mislove2009online} & Group Memberships & 88490 & 21974 & 3.24 & 12.91 \\
        \href{http://konect.cc/networks/komarix-citeseer/}{Citeseer}~\cite{kunegis2013konect} & Co-authorship & 93957 & 110371 & 5.29 & 2.99 \\
        \href{http://konect.cc/networks/dbpedia-recordlabel/}{DBPedia Labels}~\cite{kunegis2013konect,auer2007dbpedia} & Artist/Labels & 158385 & 9827 & 1.40 & 22.51 \\
        \href{http://konect.cc/networks/dbpedia-genre/}{DBPedia Genre}~\cite{kunegis2013konect,auer2007dbpedia} & Artists/Genres & 253968 & 4934 & 1.80 & 92.84 \\
        \hline
    \end{tabular}
    \caption{Hypergraph datasets: for each hypergraph, we give name, type, number of vertices, number of hyperedges, average degree, average size of hyperedge.}
    \label{tab:read_datasets}
      \vspace{-.8cm}
\end{table}

\subsection{Results}

We compare the performance of our algorithm with that of the heuristics proposed by~\citet{takaiHypergraphClusteringBased2020} for the task of approximately optimizing Problem~\ref{eq:resolvent} in the local graph partitioning setting, i.e., when the PageRank seed $s$ consists of a single vertex. We do not compare with the work of~\citet{liu2021strongly} as they focus on the slightly different strongly local graph partitioning setting, where an addition $\ell_1$-norm regularizer is added to the optimization in Problem~\ref{eq:resolvent} to enforce sparsity.

The heuristics proposed by~\citet{takai2020hypergraph} is based on discretizing the continuous-time gradient flow for Problem~\ref{eq:resolvent} by a simple forward Euler integrator. A theoretically sound implementation of this idea will in general lead to extremely small step sizes, because of the non-smooth nature of hypergraph potentials, such as $\cU_{\infty}$. \citet{takai2020hypergraph} largely ignore this issue and take uniform steps, regardless of the non-smoothness of the objective. 
The resulting heuristics is remarkably similar to the instantiation of Algorithm~\ref{alg:opt} in Theorem~\ref{thm:main1}, except that their heuristics does not specify a theoretically justified choice of step size $\eta$ and always outputs the last iterate $\hat{\vx}_T$, rather than the average iterate $\bar{\vx}_T$. Given the unspecified step size in~\cite{takai2020hypergraph}, we perform a fair comparison by evaluating the two algorithms for the same choice of step size $\eta$ as in Algorithm~\ref{alg:opt}, so that the only difference is in whether the output solution is the last iterate or the average iterate. Even with this advantageous choice for our competitor, our results show that our algorithm matches or outperforms that of~\citet{takai2020hypergraph} on all but one datasets in terms of the minimization of $\cU_\infty$ for a fixed number of iterations. On a small dataset, Network Theory, the value achieved by the heuristics of~\citet{takai2020hypergraph} tends to be near zero, leading to a large multiplicative improvement for our algorithm.
The results and setup for this comparison are described in Table~\ref{tab:takai_ratio}.

\begin{table}[ht] 
    \centering
    \begin{tabular}{c|r}
        {\bf Name}            & {\bf Median Improvement in objective value (\%)} \\
        \hline
        Network Theory        & +162.99\% \\
        Fauci Email           &   -1.75\% \\
        Arxiv Cond Mat        &   -1.60\% \\
        DBPedia Writers       &  +13.18\% \\
        DBLP                  &  +26.37\% \\
        Youtube               &  +20.98\% \\
        Citeseer              &  +25.62\% \\
        DBPedia Labels        &  +15.05\% \\
        DBPedia Genre         &  -27.52\%
    \end{tabular}
    \caption{ For each hypergraph, we consider Problem~\ref{eq:resolvent} with the $\cU_\infty$ potential for $\lambda = 0.12$ and $20$ different seed vectors $\vs$, of the form $\vs = \ve_v - \pi(\ve_v)$, where $v$ is a randomly chosen vertex and $\ve_v$ is the standard unit vector associated with $v$. We report the median multiplicative improvement in objective value of our algorithm over ~\citet{takai2020hypergraph} across the $20$ seeds.}
    \label{tab:takai_ratio}
      \vspace{-.8cm}
\end{table}

As the algorithm of~\cite{takai2020hypergraph} is very similar to Algorithm~\ref{alg:opt} in terms of computational cost, we develop a different strategy to assess the running time of our algorithm in the hypergraph setting and compare it to the running time of its closest graph analogue. On the hypergraph side, for each dataset, we consider the running time of $T=100$ iterations of our algorithm on Problem~\ref{eq:resolvent} with the standard hypergraph potential $\cU_\infty$. On the graph side, we consider the running time of computing the truncation to the first $100$ terms of the geometric series of Equation~\ref{eq:expansion} on a graph obtained by replacing each hyperedge $h$ with a complete graph over $h.$ This latter task is equivalent to Problem~\ref{eq:resolvent} where the hypergraph potential is given by the Laplacian potential of the clique expansion~\cite{takai2020hypergraph} of the hypergraph, i.e., the sum of complete graphs over the hyperedges. The results of our comparison are presented in Table~\ref{tab:timings}.

We expect the graph algorithm to run much faster than the hypergraph algorithm for two reasons. Firstly, because of the smoothness and strong convexity of the graph potential, the graph algorithm will take a much smaller number of iterations to achieve the same approximation guarantees. To highlight this property, our iterations included a stopping condition, terminating the algorithm early when a certain error threshold was met. None of the hypergraph algorithm executions triggered the condition before $T=100$ iterations, while all the graph executions did between iteration $18$ and iteration $40$, as shown in Table~\ref{tab:timings}. 

Secondly, in contrast with the hypergraph case, the computation of the action of $\mL$ in the graph case can be vectorized, leading to much faster iterations. This effect can be visualized in the last column of Table~\ref{tab:timings}, which shows the time per iteration of both algorithms for each dataset. An interesting finding here is that the vectorized iterations on the clique-expansion are approximately $10$x faster, regardless of the properties of the hypergraph.

\begin{table}[ht]
    \centering
    \begin{tabular}{c|rrr|rrr}
        {\bf Name} & \multicolumn{3}{c}{{\bf Standard Hypergraph Potential}} & \multicolumn{3}{c}{{\bf Clique Expansion}} \\
        & T & Sec & Sec/iter & T & Sec & Sec/iter \\
        \hline
        Network Theory        & 100 &     3.45 &   0.034 &  38 &     0.19 &   0.005 \\
        Fauci Email           & 100 &    17.24 &   0.171 &  35 &     0.49 &   0.014 \\
        Arxiv Cond Mat        & 100 &   177.84 &   1.761 &  39 &     6.46 &   0.161 \\
        DBPedia Writers       & 100 &   303.90 &   3.009 &  36 &     9.48 &   0.256 \\
        DBLP                  & 100 &  1021.68 &  10.116 &  32 &    21.23 &   0.643 \\
        Youtube               & 100 &   395.33 &   3.914 &  40 &    14.16 &   0.345 \\
        Citeseer              & 100 &  1303.39 &  12.905 &  32 &    43.05 &   1.305 \\
        DBPedia Labels        & 100 &   216.67 &   2.145 &  22 &     6.86 &   0.298 \\
        DBPedia Genre         & 100 &   232.77 &   2.305 &  18 &    13.36 &   0.703
    \end{tabular}
    \caption{For each hypergraph and its clique-expansion, we run 20 instances of our Algorithm~\ref{alg:opt} and the series expansion of Equation~\ref{eq:expansion} for 100 iterations. Each instance has a different vector $s$ of the form $\vs = \ve_v - \pi(\ve_v)$, where $v$ is a randomly chosen vertex and $\ve_v$ is the standard unit vector associated with $v$. 
    The value of $\lambda$ was set to 0.12. We employed the alternative stopping criterion $\|x_t - x_{t-1}\|_D^2 < 10^{-6}$ to cause early termination of our algorithm. Termination was triggered only when all 20 instances met this condition. We reported the {\it total time} and {\it total time averaged over iteration} for running all the separate 20 instances. }
    \label{tab:timings}
      \vspace{-.4cm}
\end{table}

All experiments were run on a server with a 24core Intel Xeon Silver 4116 CPU @ 2.10GHz processor and 128gb RAM. All code (except for certain NumPy operations) is single threaded and can be found in our supplementary material together with the datasets and insteuctions of how to repeat all experiments and produce all tables.

\paragraph{A Conjectured Heuristics} We consider a variant of Algorithm~\ref{alg:opt} where  $\mR$ is the graph Laplacian of the clique expansion of the instance hypergraph. We conjecture that this variation will achieve a theoretical worst-case running time that is independent of $\lambda_G$, at the cost of a potential dependence of the maximum hyperedge rank. We ran a small number of experiments on this variant, showing improvements in both accuracy and running time. We defer a full discussion of this heuristics and its experimental performace to a full version of this paper.

\newpage
\section{Details of manifold learning experiments}\label{ssec:manifold-learning-details}

In this section we give precise descriptions of each of the experiments discussed in Section~\ref{sec:experiments}. Code for performing hypergraph diffusions and implementing each of these experiments will be made publicly available upon publication. 

\subsection{Manifold learning}\label{ssec:manifold-learning-details}

In this experiment, we compared the performance of graph- versus hypergraph-based diffusions at recovering community structure in a semi-supervised setting. We consider semi-supervised problems of the following form: consider a problem where two ``communities,'' represented as compact manifolds in $\R^n$, are present and datapoints are sampled noisily from both manifolds. Given access to a small number of known community labels, can we estimate the community identity of all datapoints and distinguish the two underlying communities? We let $N$ denote the total number of sampled datapoints in these problems, and without loss of generality assume labels have the form $\{\pm 1\}^N$.

We considered three different specific instances of this problem, in which the underlying communities were (a) interlocking spirals in $\R^2$, (b) overlapping rings in $\R^2$, and (c) concentric hypersphers in $\R^5$. For each experiment, 300 datapoints were sampled uniformly at random from each community and then subjected to additive noise. Precise parametrizations of these problems and descriptions of the additive noise are given in Table~\ref{tab:manifold-learning-communities}. 

\begin{table}[ht]
\caption{Community structure and additive noise used to create problem instances for experiments in Section~\ref{sec:experiments}. For each experiment, 300 points were sampled uniformly at random from each community and then subjected to additive noise of the described magnitude.}\label{tab:manifold-learning-communities}
\begin{center}
\begin{small}
\begin{sc}
\begin{tabular}{lcccr}
\toprule
Problem instance & Community 1 & Community 2 & Additive noise \\
\midrule
Two-spirals    & 
    $\left\{\begin{bmatrix}3\theta\cos(\theta) \\
    3\theta\sin(\theta)\end{bmatrix} \mid \theta \in [\pi/2,3\pi] \right\}$
    &
    $\left\{\begin{bmatrix}-3\theta\cos(\theta) \\
    -3\theta\sin(\theta)\end{bmatrix} \mid \theta \in [\pi/2,3\pi] \right\}$
    & $\cN(0,1) $ \\[0.5cm]
\begin{tabular}[x]{@{}c@{}}Overlapping\\rings\end{tabular}   & 
    $\left\{\begin{bmatrix}2\cos(\theta) \\
    2\sin(\theta)\end{bmatrix} \mid \theta \in [0,2\pi]\right\}$
    &
    $\left\{\begin{bmatrix}3\cos(\theta) + 3\\
    3\sin(\theta)\end{bmatrix}\mid \theta \in [0,2\pi] \right\}$
    & $\cN(0,0.2) $ \\[0.5cm]
\begin{tabular}[x]{@{}c@{}}Concentric\\hyperspheres\end{tabular}   & 
    $\left\{ \vx \in \R^5 \mid \norm{\vx}_2 = 2\right\}$
    &
    $\left\{ \vx \in \R^5 \mid \norm{\vx}_2 \leq 1.3 \right\}$
    & $\cN(0,0.1) $ \\
\bottomrule
\end{tabular}
\end{sc}
\end{small}
\end{center}
\vskip -0.1in
\end{table}

\paragraph{Defining the diffusion.} We use diffusions on graphs and hypergraphs respectively to propagate the information from the known labels and estimate community identities for each node. The graphs and hypergraphs were constructed in the following manner: the 5 nearest-neighbors of each datapoint with respect to Euclidean distance were identified. This data was used to construct the 5-nearest-neighbor (5-NN) graph in which nodes correspond to datapoints and an edge is added between a datapoint and each of its 5 nearest-neighbors, and the 5-NN hypergraph in which nodes correspond to datapoints and hyperedges containing a node and all of its 5 nearest-neighbors are added.

Without loss of generality, we assign label $+1$ to datapoints sampled from community 1, and $-1$ to datapoints sampled from community 2, and denote the vector of true labels by $\vy$. As we consider the semi-supervised setting, each trial receives a small number ($5\%$ of total datapoints) of true community labels. For these experiments, 30 true labels are randomly sampled from the 600 total datapoints; note that, as a result, true labels are not necessarily distributed equally among both communities, but only uniformly sampled over the whole population. Given this set of known labels, $S$, a vector $\vx_0$ is constructed in the following manner:
\[
    \vx_0(i) = \begin{cases}
        1 \text{ if } i\in S \text{ and } \vy(i) = +1\\
        -1 \text{ if } i\in S \text{ and } \vy(i) = -1\\
        0 \text{ if } i\not\in S
\end{cases}
\]

We then use diffusions on the 5-NN graph and hypergraph to propagate information from the small random set of known labels: we initialize a discrete-time diffusion at the $\vx_0$ constructed as above and consider the sequence generated by the update rule in Equation~\ref{eq:0-s-discrete-update-rule}:
\[
    \{\vx_t | \vx_t = \vx_{t-1} -\eta \mD^{-1}\cL^\mD(\vx_{t-1})\}
\]
and the corresponding sequence generated by the graph Laplacian:
\[
    \{\vx_t | \vx_t = \vx_{t-1} -\eta\mD^{-1}\mL\vx_{t-1})\}
\]
where $\mL$ is the combinatorial graph Laplacian. In both diffusions on graphs and hypergraphs, step-size $\eta$ was chosen to be 1 in all experiments. 

\paragraph{AUC values.} For a time step $t\in \mathbb{N}$, we generate area-under-the-curve values by using $\vx_t$ to generate scores for estimated community labels in the following manner: we linearly shift and rescale entries of $\vx_t$ so that the minimum entry is 0.0 and the maximum entry is 1.0. We then use the resulting vector as the label scores, where are larger score respresents a higher prediction of belonging to community 2. We then compute the AUC value of these scores with respect to the true labels $\vy$ described above.

\paragraph{Sweep cut generation.} For a time step $t\in \mathbb{N}$, we generate estimated labels by taking a sweep-cut of $\vx_t$ with respect to a threshold $\tau$:
\[
    \vell^{\tau}_t \defeq (\vx_t > \tau)
\]
The intuition for selecting $\tau$ is the following: we want to produce a vector with low potential ($\cU(\vell)$ in the hypergraph, $\left(\vell^{\tau}_t\right)^T\mL \vell^{\tau}_t$ in the graph) that is also nearly-orthogonal to $\vone$ with respect to $\langle\cdot,\cdot\rangle_{\mD}$. In other words, we wish to find  $\vell^{\tau}_t$ with low potential that also incurs small (normalized) orthogonality penalty $\rho^\tau_t$, defined
\[
    \rho^\tau_t \defeq \langle \vell^\tau_t, \vone\rangle_{\mD}/n
\]
To do this, we perform the following procedure: first, we consider $\rho^0_t$, the penalty incurred by the sweep cut of $\vx_t$ with respect to threshold $\tau = 0$. We then consider the set of thresholds $\{\tau_i\}$ that incur orthogonality penalty similar to that of $\tau = 0$, specifically such that $\rho^{\tau_i} \leq 1.1\rho^0$. We then select
\[
    \tau \defeq \argmin_{\{\tau_i\}} \cU(\vell^{\tau_i}_t)
\]
for the hypergraph, and
\[
    \tau \defeq \argmin_{\{\tau_i\}} \left(\vell^{\tau_i}_t\right)^T\mL \vell^{\tau_i}_t
\]
for the graph.

The classification error on iteration $t$ is computed as
\[
    \textrm{error}_t \defeq \sum_{i=1}^N \mathbf{1}_{\vell_t(i) = \vy(i)}
\]
where $\mathbf{1}_{\vell_t = \vy(i)}$ is the indicator function.

\section{Code Repository}

All code, datasets and relevant instructions can be found at the following URL: \url{https://uchicago.box.com/s/qc1vxcsn52ali76kqfzn1wrybbgwiqxd}.

\end{document}